\documentclass[a4paper,11pt]{article}
\pdfoutput=1
\usepackage{jheppub}

\usepackage{amsmath,amssymb,amsthm,mathtools}
\usepackage{booktabs,array,graphicx}
\usepackage{tikz}
\usetikzlibrary{positioning}
\usepackage{float}
\usepackage{enumitem}
\usepackage{amsmath,amssymb}
\usetikzlibrary{decorations.pathmorphing,arrows.meta,calc}
 
\definecolor{bndC}{RGB}{0,114,178}    
\definecolor{bndCp}{RGB}{213,94,0}    
\definecolor{bulk}{RGB}{245,245,245}
\definecolor{edge}{RGB}{120,120,120}

\newcommand{\Z}{\mathbb{Z}}

\newcommand{\C}{\mathbb{C}}
\newcommand{\F}{\mathbb{F}}
\newcommand{\Lam}{\Lambda}
\newcommand{\Ot}{\mathrm{O}}
\newcommand{\su}{\mathfrak{su}}
\newcommand{\w}[1]{\mathtt{#1}} 

\theoremstyle{plain}
\newtheorem{theorem}{Theorem}[section]
\newtheorem{proposition}[theorem]{Proposition}
\newtheorem{lemma}[theorem]{Lemma}
\newtheorem{corollary}[theorem]{Corollary}
\theoremstyle{remark}

\numberwithin{equation}{section}

\title{TQFT gravity dual to the orbifold groupoid of Narain code CFTs over $\mathbb{Z}_4$
}

\author{Rajae Sammani and Abdellah Sebbar}
\affiliation{Department of Mathematics and Statistics, University of Ottawa,\\
150 Louis-Pasteur Private, Ottawa, ON K1N 6N5, Canada}
\emailAdd{rsammani@uottawa.ca}
\emailAdd{asebbar@uottawa.ca}

\abstract{In a companion paper \cite{A3math}, we classified the Narain conformal field theories obtained by
gluing copies of the $A_3$ root and weight lattices along their $\mathbb{Z}_4$ discriminant
group. For two hyperbolic pairs, this gives twenty-two even self-dual gluing codes, which fall
into seven inequivalent theories at $c=\bar c=6$. In this paper, we develop the physical interpretation of this family and its connections to three-dimensional topological field theory.  We show that all twenty-two codes are
topological boundary conditions of a single abelian Chern-Simons theory. Its invertible
symmetries, which include fractional $B$-field shifts, organize the boundaries into three
orbits: free, mixed and torsion. Boundaries in different
orbits are related only by gauging. We determine all gaugings of order at most four between the
seven theories, by $\mathbb{Z}_2$, $\mathbb{Z}_4$ and $\mathbb{Z}_2\times\mathbb{Z}_2$, with
explicit twisted partition functions, and follow their effect on the spectrum and the current
algebra. We then
classify the topological surface operators of the bulk theory and identify which of them implement
these gaugings. If the source and target codes have the same group type, an invertible bulk
surface relates their boundary conditions. If their group types differ, every bulk surface
relating them is non-invertible, and we determine the minimal quantum dimension of such a surface. We give the fusion rules of the condensation defects, and show
that the theories that are self-dual under gauging carry non-invertible duality defects. 
Finally, we apply the TQFT gravity sum over topologies to this family at genus one. We give a
direct proof of the handlebody sum for any number of hyperbolic pairs, and closed-form weights for
the handlebody and torsion contributions. For two pairs, the handlebody sum is the uniform average
over the free theories, and the torsion topology contributes the torsion theory, so that together
they reproduce the uniform average over all gluing codes. For three or more pairs, the handlebody
sum alone is not a positive combination of code boundary states.
}


\begin{document}
\maketitle
\flushbottom
\section{Introduction}\label{sec:intro}
Error-correcting codes provide an algebraic framework for constructing Narain lattices and their associated CFTs via Construction A. This latter lifts a code to a lattice by selecting integer vectors whose reductions modulo the alphabet size belong to the code and then rescales them. When the code is self-dual with respect to the relevant bilinear form and satisfies the corresponding evenness condition, Construction A yields an even, self-dual lattice with respect to the Lorentzian pairing $\langle (p_L,p_R),(p_L',p_R')\rangle=p_L\cdot p_L'-p_R\cdot p_R'$,
as required for a consistent Narain CFT.

The lattice vectors encode the allowed left- and right-moving momenta $(p_L,p_R)$ of the vertex operators, with conformal weights $h=p_L^2/2$ and $\bar{h}=p_R^2/2$. The code therefore
constrains which momenta occur and, through them, the CFT spectrum. The torus partition function can be expressed in terms of the associated lattice theta series or, equivalently, using the weight enumerator, and the spectral data of the CFT, such as the gap and the degeneracies of light states, can be read from the algebraic structure of the code.

This correspondence was first developed for binary quantum stabilizer codes \cite{DS21} and has since been extended in two ways. One can either enlarge the alphabet, to codes over $\Z_N$ \cite{Alam23} and over cyclotomic integers \cite{MO25}, or replace the trivial base lattice by a root lattice, so that the code glues cosets of its weight lattice \cite{SS26}. When the gluing uses the left-minus-right discriminant form, the resulting theories are the level-one WZW models \cite{Ang25}. Because the code is a finite group, the symmetries of these theories and their gaugings can also be described in finite terms \cite{Kawabata:2023iss,AKN26,IMP26}.

Codes are useful because they replace continuous data with finite data, Narain moduli space is continuous, whereas the codes with a fixed alphabet and length form a finite set. Enumerating these codes therefore provides a finite classification of the associated theories. These constructions can identify rational points in moduli space, including points of enhanced symmetry and theories with large spectral gaps, which may be difficult to locate directly in the continuous moduli space. The code also provides direct access to the spectrum. 

From the bulk perspective, the average of Narain theories over moduli space, which is dual to the sum over handlebodies in $U(1)$ gravity \cite{MW20,AJCHT21}, also admits a discrete interpretation in terms of codes. Replacing the continuous moduli space with a finite set of code-based theories leads to an abelian Chern-Simons theory summed over topologies,
often referred to as TQFT gravity
\cite{Rae21,ADS23,Barbar:2023ncl,Dymarsky,DHM25,Bar25,Angelinos:2025tqft,Saidi:2024zam, Saidi:2026code}. The code theories then arise as topological boundary conditions of a single bulk theory. When $N$ is prime or square-free, the handlebody sum reproduces the uniform average over these theories \cite{ADS23,Dymarsky,DHM25}.

The split momentum-winding gluing of root lattices was introduced in \cite{MO25}, including
codes over $\Z_q$ with non-prime $q$. There, the codes have generator matrices of the form
$(I\,|\,B)$ and the Narain data are $G=I\otimes C^{-1}$, $B=B\otimes C^{-1}$; such codes are always free. When $q$ is not square-free, other Lagrangian codes exist, which this description does not encompass. On the holographic side, it is known that for non-square-free $N$ the handlebody sum alone does not reproduce the uniform average over all code theories, and topologies with line defects must be included \cite{Dymarsky,Angelinos:2025tqft}. For $N=4$ this was shown explicitly by Dymarsky and Shapere \cite{Dymarsky}, who identified the bulk theory as $\Z_4$ gauge theory, derived from first principles that the boundary ensemble is uniform, and computed the genus-one bulk sum for one, two and three pairs. The genus-one sum for $N=p^2$ and any number of pairs was then obtained in \cite[appendix~D.3]{BDS25}, and the odd prime powers were treated in \cite{Angelinos:2025tqft}. Additionally, the simplest family in which non-free codes appear is the
$A_3$ family over $\Z_4$, which is the subject of the companion paper \cite{A3math}. There, we arranged copies of $A_3$ in $d$ hyperbolic pairs, and the labels are glued with the split pairing
$\beta((a,b),(a',b'))=a\cdot b'+b\cdot a'$ on $\Z_4^d\oplus\Z_4^d$. An even self-dual
gluing code then defines a Narain CFT with $c=\bar c=3d$.

The first result of \cite{A3math} identifies the ordinary $\Z_4$ codes \cite{Alam23} that
arise from the $A_3$ gluing construction. The code
$E=\{000,220,202,022\}\subset\Z_4^3$ is twice the binary even-weight code of length three.
Its normalized Construction A lattice is $D_3\simeq A_3$, and $E^\perp/E\simeq\Z_4$
reproduces the discriminant group. Replacing each gluing label $a$ by the coset $ag+E$, with
$g=(1,1,1)$, turns every $A_3$ gluing code into an ordinary even self-dual code over $\Z_4$
of length $6d$. The two lattices coincide as subsets of momentum space, and the
identification preserves both the Lorentzian pairing and the positive energy form. The
ordinary codes obtained in this way are exactly those that contain $E^{\oplus 2d}$ and lie
in $(E^\perp)^{\oplus 2d}$.

The second result is a complete classification for one and two hyperbolic pairs. For $d=1$
there are three gluing codes and two block-monomial classes, at $c=3$. For $d=2$ there are
twenty-two codes. They are counted by reducing modulo two onto the split quadratic space
$\F_2^4$, which can be identified with $2\times2$ matrices with the determinant as quadratic
form. In these terms, the six totally singular planes give twelve codes of type $\Z_4^2$, the
nine nonzero singular vectors give nine codes of type $\Z_4\oplus\Z_2^2$, and the zero
subspace gives one code of type $\Z_2^4$. Under block-monomial equivalence the twenty-two
codes fall into seven classes, $F_0,F_1,F_2,M_0,M_1,M_2$ and $T$, with orbit sizes
$4,4,4,4,4,1,1$. For each class, \cite{A3math} gives an integral null basis, explicit
metric and $B$-field data, and the exact theta series. The seven theta series are
pairwise distinct, so the seven classes are seven distinct $\Ot(6,6;\Z)$ classes at
$c=\bar c=6$.

Several structural results go beyond the classification. First, for every $d$ the chiral
root system on each side is $A_1^{3m}$, where $m$ counts the words $(2e_i,2e_i)$ in the code.
The current algebra is therefore $\widehat{\su}(2)_1^{3m}\oplus\widehat{\mathfrak u}(1)^{3d-3m}$,
and no member of the family has an $\su(4)$ current algebra. This separates the split family
from the level-one $\mathrm{SU}(4)$ WZW models of \cite{Ang25}, whose discriminant form
$3(a\cdot a'-b\cdot b')$ is not isometric to $\beta$ over $\Z_4$. Second, $M_2$ is the only
two-pair class with no T-duality frame of zero $B$-field. Its positive lattice is
$D_{12}^+$, and the irreducible $D_{12}$ root system spanning all twelve directions forbids
any orthogonal splitting $G^{-1}\oplus G$. In the frame of \cite{A3math} its $B$-field is
half-integral and cannot be made integral. Third, the finite orthogonal group of
$(\Z_4^4,q)$, which includes fractional $B$-field shifts, merges the seven classes into only
three orbits, labeled by the group type of the code, free, mixed and torsion. These finite orbits are therefore not
T-duality classes. Finally, any two gluing codes $C,C'$ are related by a shift orbifold with
gauge group $C'/(C\cap C')$ \cite{IMP26}. The target code fixes the twisted-sector
projections, for example, gauging
$(-1)^{2\Delta}$ in the torsion class $T$ produces the second mixed class $M_2$.

In this paper we develop the physical interpretation of this family of CFTs, our results can be summarized as follows:
\begin{enumerate}[label=(\roman*)]
    \item First, we identify the corresponding bulk theory. Every code lattice contains the zero-code lattice $R^d\oplus R^d$, whose discriminant form is the split module $(\Z_4^{2d},\tfrac34\,a\cdot b)$. This defines a single abelian Chern-Simons theory $\mathcal Z_d$, equivalent to an untwisted $\Z_4^d$ gauge theory. Its topological boundary conditions are exactly the even self-dual gluing codes, so all twenty-two theories at $d=2$ are boundaries of the same bulk, in the sense of the SymTFT construction \cite{KS11,GK21,KOZ22,FMT22}. Upon considering invertible bulk symmetries in the form of block-monomial moves, which are T-dualities, and fractional $B$-field shifts, which describe jumps between different CFTs, we organize the twenty-two boundaries into three orbits, one for each group type of the code. Since an invertible symmetry cannot change the group type, the free, mixed and torsion theories can only be connected by gaugings.
    \item Second, we construct the complete web of gaugings between the seven theories. We derive the order and structure of each gauge group from the reductions of the two codes modulo two, and determine every connection by a group of order at most four, $\Z_2$, $\Z_4$ or $\Z_2\times\Z_2$. Three pairs of classes are joined only by gaugings of order eight, and five of the seven classes are self-dual under some gauging. For each connection we give the twisted partition functions and follow the light spectrum and the $\widehat{\su}(2)_1$ currents. For example, the theory $M_2$, with the largest gap, is obtained from $T$ by gauging $(-1)^{2\Delta}$, which removes all states of half-integral dimension while keeping the $\widehat{\su}(2)_1^6$ currents.    \item Third, we classify the topological surface operators of $\mathcal Z_2$ and determine
    which gaugings they implement. By the folding trick, the surfaces of the bulk theory are
    labelled by gluing codes with twice as many pairs, and the invertible ones are the T-dualities
    and fractional $B$-field shifts. A surface acts on the boundaries by condensing some lines and
    transmitting the others. We show that if the source and target codes have the same group
    type, an invertible bulk surface relates their boundary conditions, while if their group types
    differ, every bulk surface relating them is non-invertible. The minimal quantum dimension of
    such a surface is $2^{|j(C)-j(C')|}$, where $j$ counts the $\Z_4$ factors of the code. We give the fusion rules of the
    condensation defects that implement the gaugings. On the boundary, the theories that are
    self-dual under gauging carry non-invertible duality defects of Tambara-Yamagami type
    \cite{TY98}, and we determine the symmetric bicharacter of each.
        \item Fourth, we apply the TQFT gravity sum over topologies of \cite{Dymarsky} to this family.
    At genus one it receives two contributions, the handlebodies and a torsion topology. We give a
    direct proof of the handlebody sum for any number of pairs, using the Weil representation, and
    we show that the weights of the two contributions, known from \cite{Dymarsky,BDS25} as sums
    over orbits of codes, take a simple closed form. The $A_3$ embedding then identifies each
    contribution with explicit CFTs. For $d=2$ the handlebody sum is the average of the free
    theories $F_0,F_1,F_2$ alone, and the torsion topology contributes the torsion theory $T$, so
    that $Z_{\rm all}=\tfrac9{11}Z_{\rm free}+\tfrac2{11}Z_T$. The handlebody state admits a
    positive expansion using only free codes. The mixed-code sum satisfies
    $V_M=\tfrac12V_F+3V_T$, so the contribution of the mixed theories, including $M_2$ with the
    largest gap, to the full ensemble (here $V_F$, $V_M$, $V_T$ are the sums of the free, mixed
    and torsion code vectors) is absorbed into both the handlebody and the torsion terms.
    For $d\geq3$ we show that the handlebody state alone is not a nonnegative combination of code
    boundary vectors, which makes concrete in this family the statement of
    \cite{Dymarsky} that handlebodies alone do not give a consistent bulk sum.
\end{enumerate}

The paper is organized as follows. In section~\ref{sec:setup} we review the $A_3$ family and
its seven T-duality classes. In section~\ref{sec:bulk} we recall the bulk Chern-Simons
theory and show how its invertible symmetries organize the boundary conditions. In
section~\ref{sec:web} we construct the orbifold web between the seven classes and analyze how
each gauging acts on the spectrum. In section~\ref{sec:surfaces} we classify the surface operators
of the bulk, determine which of them implement the gaugings, give their fusion rules and the
duality defects of the boundary theories. In section~\ref{sec:ensemble} we specialize the sum over
topologies of \cite{Dymarsky} to this family. We conclude in section~\ref{sec:disc} with a discussion of our results and possible future directions. In the appendices, Appendix~\ref{app:web-proofs}, Appendix~\ref{app:surfaces} and Appendix~\ref{app:handlebody}, we collect the proofs and computational details to avoid cluttering the main text.

\section{The $A_3$ Narain family over $\Z_4$}\label{sec:setup}

We recall the construction of~\cite{A3math}. Let $R=A_3\simeq D_3$ be the root lattice and $W=R^*$ its weight
lattice~\cite{CS99}. The quotient $W/R\simeq\Z_4$ has four cosets $W_a$, $a\in\Z_4$, and we write
$[x]=a$ when $x\in W_a$. Inner products of weights are fixed modulo one by their labels,
\begin{equation}\label{eq:disc-pairing}
  x\cdot y\equiv\tfrac34\,ab \pmod 1,\qquad x\in W_a,\ y\in W_b .
\end{equation}
Now, we take $d$ pairs of $A_3$ blocks. The labels form the group
$M_d=\Z_4^d\oplus\Z_4^d$, whose elements codewords $c=(a,b)=(a_1,\dots,a_d,b_1,\dots,b_d)$.
For $d=2$ this is $(a_1,a_2,b_1,b_2)$, often abbreviated $a_1a_2b_1b_2$. Following~\eqref{eq:disc-pairing}, we equip $M_d$ with the integer-valued forms
\begin{equation}\label{eq:beta-q}
  \beta\big((a,b),(a',b')\big)=a\cdot b'+b\cdot a' \pmod 4,\qquad q(a,b)=a\cdot b\pmod 4 .
\end{equation}
A code is thus a subgroup of $M_d$,  $C\subset M_d$, it is even self-dual when $q=0$ and $C=C^\perp$ with respect
to $\beta$. It selects which cosets are allowed in the momentum lattice,
\begin{equation}\label{eq:Lambda-C}
  \Lam(C)=\big\{(x,y)\in W^d\oplus W^d:\ ([x_1],\dots,[x_d],[y_1],\dots,[y_d])\in C\big\}.
\end{equation}
On $\Lam(C)$, we use the Lorentzian form and the energy form, given respectively by:
\begin{equation}
    \eta(z,z')=x\cdot y'+y\cdot x' \hspace{0.5cm},  \hspace{0.5cm}
    \mathcal H(z,z')=x\cdot x'+y\cdot y'
\end{equation}
for $z=(x,y)$ and $z'=(x',y')$. The left and right momenta are
$p_{L,R}=(x\pm y)/\sqrt2$, and a momentum primary with lattice vector $z$ has
\begin{equation}
  h-\bar h=\tfrac12\eta(z,z)=x\cdot y,\qquad \Delta=h+\bar h=\tfrac12 \mathcal H(z,z).
\end{equation}
By~\eqref{eq:disc-pairing}, the Lorentzian form reduces to the code forms such that
\[
  \eta(z,z')\equiv\tfrac34\,\beta(c,c'),\qquad \tfrac12\eta(z,z)\equiv\tfrac34\,q(c)\pmod 1 .
\]
Therefore, $\Lam(C)$ is even unimodular for $\eta$, and defines a consistent CFT with $c=\bar c=3d$, exactly
when $C$ is even self-dual~\cite{A3math}.

The torus partition function is built from the block theta functions
$\Theta_{ab}$ of~\cite{A3math}, one for each pair of labels. For a function
$f:M_d\to\C$ on labels we write

\begin{equation}\label{eq:Theta-f}
  \Theta[f]=\sum_{c\in M_d}f(c)\prod_{i=1}^d\Theta_{a_ib_i},\qquad
  Z_C=\frac{\Theta[\mathbf 1_C]}{|\eta(\tau)|^{6d}},
\end{equation}
where $\mathbf 1_C$ is the indicator function of $C$.

certain codes give the same theory and are related by transformations like the block-monomial moves. These are
generated by the permutation $\sigma$ of the two hyperbolic pairs, the exchanges
$\varepsilon_i:a_i\leftrightarrow b_i$ and the sign changes $\nu_i:(a_i,b_i)\mapsto(-a_i,-b_i)$, which induce T-dualities. For $d=2$ there are twenty-two even self-dual codes, which fall into seven
block-monomial classes~\cite{A3math}. Their theta series are all different, so the seven
classes are seven distinct Narain theories. They are listed in
table~\ref{tab:seven}. 
\begin{table}[]
\centering
\renewcommand{\arraystretch}{1.15}
\begin{tabular}{llcccc}
\hline
Class & representative $C_X$ & group type & orbit size & gap & chiral roots \\
\hline
$\mathrm F_0$ & $\langle 0001,\ 0010\rangle$               & $\mathbb Z_4^2$                 & 4 & $3/8$ & --- \\
$\mathrm F_1$ & $\langle 0011,\ 1300\rangle$               & $\mathbb Z_4^2$                 & 4 & $3/4$ & --- \\
$\mathrm F_2$ & $\langle 0012,\ 2100\rangle$               & $\mathbb Z_4^2$                 & 4 & $1/2$ & --- \\
$\mathrm M_0$ & $\langle 0001,\ 0020,\ 2000\rangle$        & $\mathbb Z_4\oplus\mathbb Z_2^2$ & 4 & $3/8$ & $A_1^3$ \\
$\mathrm M_1$ & $\langle 0002,\ 0011,\ 2200\rangle$        & $\mathbb Z_4\oplus\mathbb Z_2^2$ & 4 & $1/2$ & --- \\
$\mathrm M_2$ & $\langle 0022,\ 0202,\ 1113\rangle$        & $\mathbb Z_4\oplus\mathbb Z_2^2$ & 1 & $1$   & $A_1^6$ \\
$\mathrm T$   & $\langle 0002,\ 0020,\ 0200,\ 2000\rangle$ & $\mathbb Z_2^4$                 & 1 & $1/2$ & $A_1^6$ \\
\hline
\end{tabular}
\caption{The seven block-monomial classes at $c=\bar c=6$~\cite{A3math}.
The second column gives the representative code $C_X$ of each class, with
labels ordered as $(a_1,a_2,b_1,b_2)$. The orbit size is the
number of labeled codes in the class. The gap is the lowest nonzero
$\Delta$ of a momentum primary. The last column gives the chiral root system.}
\label{tab:seven}
\end{table}
\section{The bulk theory and its invertible symmetries}\label{sec:bulk}
In this section we recall the 3D bulk theory whose topological boundary conditions
are the code theories of section~\ref{sec:setup}, and its invertible symmetries. The bulk theory is
level-four $AB$ theory, that is $\Z_4^d$ gauge theory, whose boundary conditions and sum over
topologies were studied in \cite{ADS23,Barbar:2023ncl,Dymarsky}. We review it in the
conventions of the $A_3$ family. The construction uses only the lattice $\Lam_0=R^d\oplus R^d$ of
the zero code, which is contained in every code lattice $\Lam(C)$. Its discriminant group with
respect to $\eta$ is the group $M_d$ of code labels, and its discriminant form defines an abelian
Chern-Simons theory $\mathcal Z_d$, which is a symmetry topological field theory in the sense of
\cite{GK21,KOZ22,FMT22}. As shown in \cite{Barbar:2023ncl,Dymarsky} for $AB$ theory, the
topological boundary conditions of $\mathcal Z_d$ are labeled by the even self-dual codes, so that
every code theory is obtained from the same bulk by a choice of boundary. We then turn to the
invertible symmetries of $\mathcal Z_d$, which act on its boundaries. For $d=2$ they arrange the
twenty-two boundaries into three orbits, one for each group type of the code. The bulk theory
alone does not single out which of these symmetries are dualities of the boundary CFTs, since
this depends on the relative theory $\Psi$, that is, on the Narain point. The $A_3$ construction
fixes $\Psi$ explicitly. This allows us to separate the invertible symmetries into T-dualities,
which preserve $\Psi$ and relate equivalent CFTs, and fractional $B$-field shifts, which change
$\Psi$ and relate different CFTs.

\subsection{The bulk TQFT}

Every code lattice contains the associated lattice of the zero case $\Lam_0=R^d\oplus R^d$. Particularly, with respect to
$\eta$, we can write
\begin{equation}\label{eq:disc-form}
  \Lam_0^{*}/\Lam_0 = (W^d\oplus W^d) / (R^d\oplus R^d) \;\simeq\; M_d,
  \qquad
  \mathsf q(a,b)=\tfrac12\eta(z,z)\bmod 1 = \tfrac34\, a\cdot b \bmod 1 ,
\end{equation}
where $z$ is any lift of $(a,b)\in M_d$. The pair $(M_d,\mathsf q)$ is the anyon data of the abelian
Chern-Simons theory $\mathcal Z_d$ whose $K$-matrix is the $\eta$-Gram matrix of $\Lam_0$,
\begin{equation}\label{eq:bulk}
  S_{\mathcal Z_d}=\frac{1}{4\pi}\int K_{IJ}\,A^I\wedge dA^J,\qquad
  K=\begin{pmatrix}0 & K_{A_3}\otimes I_d\\ K_{A_3}\otimes I_d & 0\end{pmatrix},
\end{equation}
with $K_{A_3}$ the Cartan matrix of $A_3$. Setting $A=(a,b)$, the Chern-Simons theory can be equivalently expressed as the BF theory given by $\frac1{2\pi}\int b^{\mathsf T}(K_{A_3}\otimes I_d)\,da$.
Its lines are labeled by $M_d$, which has order $|\det K|=4^{2d}$, and the spin of the line $(a,b)$
is $\mathsf q(a,b)$. Since $K$ has signature $(3d,3d)$, the chiral central charge vanishes and
$\mathcal Z_d$ admits gapped boundary conditions. The map $(a,b)\mapsto(a,-b)$ is an isometry from
$(M_d,\mathsf q)$ to $(\Z_4^{2d},\tfrac14\,a\cdot b)$, which is the discriminant form of $d$ copies of
level-four $AB$ theory $\frac{4}{2\pi}\int B\wedge dA$ \cite{ADS23,Dymarsky}. Since the line
content and its braiding determine an abelian Chern-Simons theory up to decoupled invertible
factors, $\mathcal Z_d$ is $\Z_4^d$ gauge theory. In particular the counts $3$, $22$ and $486$ of
topological boundary conditions for $d=1,2,3$ agree with \cite[section~2.4]{Dymarsky}.

The relative theory with momentum lattice $\Lam_0$ has the $4^{2d}$ partition vectors
$\Theta_{(a,b)}/|\eta(\tau)|^{6d}$, where $\Theta_{(a,b)}=\prod_{i=1}^d\Theta_{a_ib_i}$ is a product of
the block theta functions of~\cite{A3math}. It lives on one boundary of the slab
$\mathcal Z_d\times[0,1]$. A topological boundary condition on the other side is specified by a
Lagrangian subgroup $L\subset M_d$, that is, $\mathsf q|_L=0$ and $|L|^2=|M_d|$~\cite{KS11}. Since three
is invertible modulo four, $\mathsf q|_L=0$ is the evenness condition $a\cdot b\equiv0\pmod4$ on $L$.
By polarization an even subgroup is isotropic, and an isotropic subgroup of order $4^d$ is
self-dual~\cite{A3math}. Hence the Lagrangian subgroups of $(M_d,\mathsf q)$ are exactly the
even self-dual codes. Shrinking the slab with boundary $L$ projects the partition vector onto the labels
in $L$ and produces the absolute theory with momentum lattice $\Lam(L)$. This is the SymTFT, or sandwich,
picture~\cite{GK21,KOZ22,FMT22} specialized to the present family. For $d=2$, the classification
of \cite{A3math} therefore gives exactly twenty-two topological boundary conditions of
$\mathcal Z_2$, forming the seven classes of table~\ref{tab:seven}.

\subsection{Bulk symmetries and dualities}\label{sec:bulk-sym}
 
The lines of $\mathcal Z_d$ are labelled by $M_d$, and a topological boundary $L$ is the set of lines
that can end on it. An invertible symmetry of $\mathcal Z_d$ relabels the lines without changing
their spins and braiding, so it is an element $g$ of the isometry group $\Ot(M_d,\mathsf q)$. Acting
on a boundary, it changes which lines end there, $L\mapsto gL$. In the present family the bulk
symmetries come in two kinds, distinguished by what they do to the Narain background.
 
\paragraph{T-dualities:} Permuting the hyperbolic pairs, exchanging $a_i\leftrightarrow b_i$ and
reversing the sign of a pair are the block-monomial operations of~\cite{A3math}. They are induced by
$\Ot(3d,3d;\Z)$ transformations that preserve both $\eta$ and $\mathcal H$. The theories with boundaries
$L$ and $gL$ are then the same CFT.
 
\paragraph{Fractional $B$-field shifts:} Shifting the $B$-field by an integral antisymmetric matrix is
also a T-duality. Shifting it by a fractional amount is not: it moves the theory to a different point of
the Narain moduli space. In the present family the relevant shift is $B\mapsto B+sJ\otimes K_{A_3}^{-1}$,
with $s\in\Z_4$ and $J=\left(\begin{smallmatrix}0&1\\-1&0\end{smallmatrix}\right)$. On the code labels it acts as
\begin{equation}\label{eq:B-shift}
  h_s:(a,b)\longmapsto(a,\,b+sJa)\qquad\text{on }M_2 .
\end{equation}
This is the familiar effect of a $B$-field. In the null coordinates of~\cite{A3math}, where
$a$ labels windings and $b$ momenta, a state with winding label $a$ acquires an additional momentum
label $sJa$. Because $J$ is antisymmetric, $a\cdot Ja=0$, so the spin $h-\bar h$, which is
the product of momentum and winding, does not change. The map $h_s$ is therefore a symmetry of the
bulk. The scaling dimensions $\Delta$ do change, because $h_s$ does not preserve $\mathcal H$.
 
To illustrate this, let us consider the free classes. By~\cite{A3math}, the representatives
$C_0\in F_0$, $C_1\in F_1$ and $C_2\in F_2$ have the same metric $G=I_2\otimes K_{A_3}^{-1}$ and the
$B$-fields $B=bJ\otimes K_{A_3}^{-1}$ with $b=0,1,2$. The three free classes thus correspond to same
torus at three values of the $B$-field, and $h_s$ maps $C_b$ to $C_{b+s}$. Their gaps however, $3/8$, $3/4$ and
$1/2$, show that these are three different CFTs.
 
Combining the two kinds of symmetries organizes all boundaries. Under $\Ot(M_2,\mathsf q)$, the twenty-two boundaries of $\mathcal Z_2$ form exactly three orbits,
\[
  \{F_0,F_1,F_2\}\ (12\ \text{codes}),\qquad \{M_0,M_1,M_2\}\ (9\ \text{codes}),\qquad \{T\}\ (1\ \text{code}).
\]
 
The orbits are labeled by the group type of the code, $\Z_4^2$, $\Z_4\oplus\Z_2^2$ or $\Z_2^4$, as
for $\Z_4$ codes of any length \cite[section~2.3]{Dymarsky}. We can understand this as how many of the lines ending on the boundary have order four
and how many have order two. A relabeling of lines cannot change it. Within each type, T-dualities
and the $B$-field shifts~\eqref{eq:B-shift} connect all codes.

\section{The seven theories and the orbifolding web}\label{sec:web}
 
In section~\ref{sec:bulk} we saw that the $22$ code theories are the topological boundary conditions of a single bulk theory $\mathcal Z_2$. Two boundaries of the same abelian Chern-Simons theory are always related by gauging a finite symmetry of either boundary theory~\cite{KS11,GK21}. Gaugings of $\Z_N$ symmetries of Narain and code CFTs were studied in~\cite{AKN26}, with orbifold partition functions written through weight enumerators. Here we make this explicit for the $A_3$ family. We derive the orbifold partition functions, determine exactly which pairs of the seven classes are related by gauging a group of order two or four, and describe what each gauging does to the spectrum. The mathematical proofs are collected in appendix~\ref{app:web-proofs}.

\subsection{Gauging between boundaries}\label{sec:gauging}
 
The bulk lines with labels in $C$ end on the boundary $C$. The remaining bulk lines, pushed onto the boundary and taken modulo $C$, become topological lines of $\Lambda(C)$. They generate the finite abelian symmetry $M_d/C$ of $\Lambda(C)$. A label $w\in M_d$ acts on a vertex operator $V_z$ of label $c$ by the phase
\begin{equation}\label{eq:character}
  \chi_w(c)=e^{2\pi i\,\frac34\beta(c,w)}=e^{2\pi i\,\eta(z,\hat w)} ,
\end{equation}
where $\hat w$ is any lift of $w$ to $W^d\oplus W^d$. So $\chi_w$ is the shift symmetry by $\hat w$~\cite{IMP26}. It depends only on $w$ modulo $C$, since $\beta(c,C)=0$ for $c\in C$. The second equality follows from $x\cdot y\equiv 3ab/4 \pmod 1$ for $x\in W_a$, $y\in W_b$.
 
Let $C'$ be a second boundary and put
\begin{equation}\label{eq:I-H}
  I=C\cap C',\qquad H=C'/I .
\end{equation}
The pairing $\beta$ identifies $C/I$ with the dual of $C'/I$, so the two groups are isomorphic and the relation below is symmetric in $C$ and $C'$. For $w,w'\in C'$ we define the partition function with a $w$-twist along space and a $w'$-insertion along time,
\begin{equation}\label{eq:twisted-Z}
  Z_C[w,w']=\frac{1}{|\eta(\tau)|^{6d}}\sum_{c\in C+w}\chi_{w'}(c)\prod_{i=1}^d\Theta_{a_ib_i}.
\end{equation}
It depends only on the classes of $w,w'$ in $H$. Because $C'$ is even and isotropic, these twisted partition functions transform covariantly (appendix~\ref{app:modular}):
\begin{equation}\label{eq:twisted-modular}
  Z_C[w,w'](\tau+1)=Z_C[w,w'+w],\qquad Z_C[w,w'](-1/\tau)=Z_C[w',-w].
\end{equation}
Evenness, $q(w)=0$, gives the spin selection rule in the twisted sectors: the spins of the $w$-twisted states are fixed modulo one by their $H$-charge. Isotropy, $\beta(w,w')=0$, is the absence of a mixed 't~Hooft anomaly. Together, evenness and isotropy of $C'$ say that $H$ is non-anomalous. The orbifold by $H$ is then
\begin{equation}\label{eq:orbifold-sum}
  \frac{1}{|H|}\sum_{w,w'\in H}Z_C[w,w']=Z_{C'} .
\end{equation}
Projecting onto $H$-invariant states keeps exactly the labels $I+w$ in the $w$-twisted sector, and these cosets assemble into $C'$. Thus the untwisted invariant states form $\Lambda(I)$, the $w$-twisted states form $\Lambda(I+w)$, and the orbifold is $\Lambda(C')$.
 
\subsection{An example: $M_2=T/\mathbb Z_2$}\label{sec:example}
 
The code $T=2\mathbb Z_4^4$ consists of the $16$ words with entries in $\{0,2\}$. In $\Lambda(T)$ every block lies in $W_0\cup W_2=\mathbb Z^3$, so $\Lambda(T)$ is $\mathbb Z^{12}$ with its standard positive form. The words of $C_{M_2}$ are $(k,k+2b,k+2a,3k+2a+2b)$ with $k\in\mathbb Z_4$ and $a,b\in\mathbb Z_2$. Those with $k$ even are twice the binary even-weight words; they form $I=T\cap C_{M_2}$, of order eight. Hence $H\simeq\mathbb Z_2$, generated by $w=1113$.
 
On $T$ one has $\beta(c,1113)=a_1+3a_2+b_1+b_2$, so $\chi_w(c)=(-1)^{N_2(c)}$, with $N_2(c)$ the number of entries equal to two. A block of $\mathbb Z^3$ has label two exactly when its coordinate sum is odd. Therefore
\begin{equation}\label{eq:minus-2Delta}
  \chi_w(z)=(-1)^{z_1+\dots+z_{12}}=(-1)^{2\Delta(z)} ,
\end{equation}
and, since oscillators shift $\Delta$ by integers, $w$ acts on all states as $(-1)^{2\Delta}$. Gauging it removes the states of half-integral dimension. The invariant sector is $D_{12}$, the twisted sector is $(\mathbb Z+\tfrac12)^{12}$, and the projected twisted sector is one spinor coset of $D_{12}$. Thus $\Lambda(C_{M_2})=D_{12}^+=D_{12}\cup(s+D_{12})$, with $s=\tfrac12(1,\dots,1,-1)$, in agreement with~\cite{A3math}.
 
Pairing coordinates as $(a_1,b_1)$, $(a_2,b_2)$ and using $\Theta_{ab}=\Theta_{ba}=\Theta_{-a,-b}$, put
\begin{equation}
  P=\Theta_{00}+2\Theta_{02}+\Theta_{22},\quad Q=\Theta_{00}-2\Theta_{02}+\Theta_{22},\quad
  R=2(\Theta_{11}+\Theta_{13}),\quad S=2(\Theta_{11}-\Theta_{13}).
\end{equation}
The four sectors are
\begin{equation}\label{eq:T-M2-sectors}
  |\eta|^{12}Z_T[0,0]=P^2,\quad |\eta|^{12}Z_T[0,w]=Q^2,\quad |\eta|^{12}Z_T[w,0]=R^2,\quad |\eta|^{12}Z_T[w,w]=-S^2 .
\end{equation}
Under $\tau\to\tau+1$, $\Theta_{11}\to-i\Theta_{11}$ and $\Theta_{13}\to i\Theta_{13}$, so $R^2\to-S^2$, in agreement with \eqref{eq:twisted-modular}. The orbifold sum gives
\begin{equation}\label{eq:Z-M2}
  Z_{M_2}=\frac{P^2+Q^2+R^2-S^2}{2|\eta|^{12}}=\frac{(\Theta_{00}+\Theta_{22})^2+4\Theta_{02}^2+8\,\Theta_{11}\Theta_{13}}{|\eta|^{12}} .
\end{equation}
The last term is the twisted sector. Its lightest states have $\Delta=3/2$. Physically, the gauging keeps the $\widehat{\su}(2)_1^6$ current algebra of $T$, which is neutral under $(-1)^{2\Delta}$, and raises the gap from $1/2$ to $1$: the $24$ lightest states of $T$ are odd and are projected out, and the twisted sector starts at $\Delta=3/2$.
 
\subsection{The complete web}\label{sec:web-result}
 
We now determine every gauging by a group of order at most four. The order and type of $C/(C\cap C')$ are fixed by the reductions of $C$ and $C'$ modulo two. If these reductions span subspaces $W_C,W_{C'}\subset\mathbb F_2^4$ of dimensions $\omega,\omega'$ that meet in dimension $k\le1$, then
\begin{equation}\label{eq:order-formula}
  |C/(C\cap C')|=2^{\,\omega+\omega'-2k}.
\end{equation}
An order-four quotient is $\mathbb Z_2\times\mathbb Z_2$ when $W_C$ and $W_{C'}$ are orthogonal, and $\mathbb Z_4$ otherwise. Distinct codes over the same plane ($k=2$) are always related by $\mathbb Z_2\times\mathbb Z_2$. These statements are proved in appendix~\ref{app:reduction}. Running through the reductions of the seven classes (appendix~\ref{app:web-proof}) gives the web of table~\ref{tab:web} and figure~\ref{fig:web}.
 
\begin{table}[t]
\centering
\small
\renewcommand{\arraystretch}{1.25}
\begin{tabular}{@{}lll@{}}
\toprule
Gauge group & Connections between distinct classes & Self-connections \\
\midrule
$\mathbb Z_2$ & $F_0M_0,\ F_0M_1,\ F_1M_1,\ F_1M_2,\ F_2M_0,\ F_2M_1,\ M_0T,\ M_1T,\ M_2T$ & --- \\
$\mathbb Z_4$ & $F_0F_1,\ F_0F_2,\ F_1F_2,\ M_0M_1,\ M_0M_2$ & $F_0,\ F_1,\ F_2,\ M_0,\ M_1$ \\
$\mathbb Z_2\times\mathbb Z_2$ & $F_0F_2,\ F_0T,\ F_1T,\ F_2T,\ M_0M_1,\ M_1M_2$ & $F_1,\ M_0,\ M_1$ \\
\midrule
order $8$ only & $F_1M_0,\ F_0M_2,\ F_2M_2$ & \\
\bottomrule
\end{tabular}
\caption{The orbifolding web. A connection $XY$ means that some $C\in X$ and $C'\in Y$ are related by gauging the group of that row. A self-connection means the same with $X=Y$ and $C\neq C'$.}
\label{tab:web}
\end{table}
 
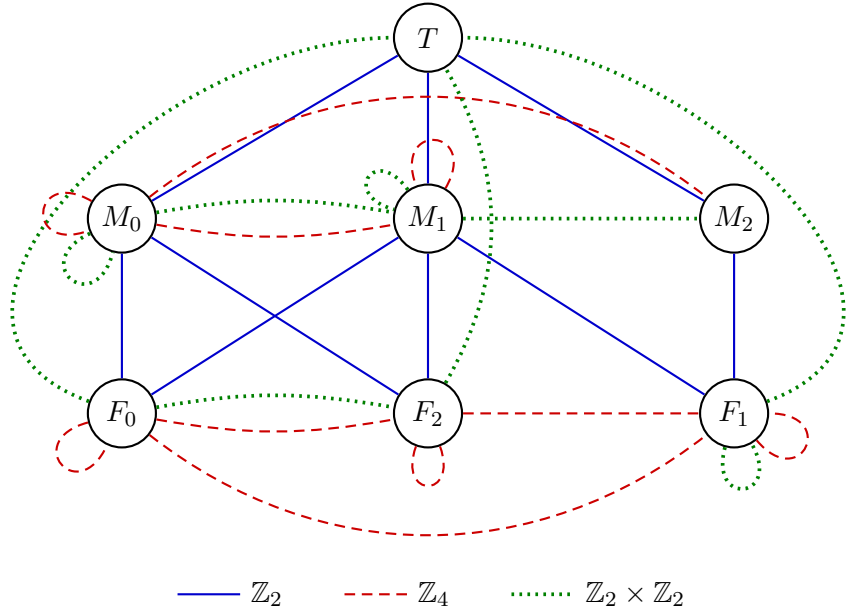
\begin{figure}[t]
\centering
\begin{tikzpicture}[scale=0.92,
  cls/.style={circle,draw,thick,minimum size=9mm,inner sep=0pt,fill=white},
  z2/.style={thick,blue!80!black},
  z4/.style={thick,red!80!black,dash pattern=on 4pt off 2.5pt},
  zz/.style={line width=1.4pt,green!50!black,dash pattern=on 1pt off 2pt},
  selfloop/.style={looseness=6}]
  \coordinate (pT) at (0,5.4);
  \node[cls] (T)  at (0,5.4) {$T$};
  \node[cls] (M0) at (-4.4,2.8) {$M_0$};
  \node[cls] (M1) at (0,2.8)  {$M_1$};
  \node[cls] (M2) at (4.4,2.8) {$M_2$};
  \node[cls] (F0) at (-4.4,0) {$F_0$};
  \node[cls] (F2) at (0,0)    {$F_2$};
  \node[cls] (F1) at (4.4,0)  {$F_1$};
  \draw[z2] (F0)--(M0); \draw[z2] (F0)--(M1); \draw[z2] (F2)--(M0);
  \draw[z2] (F2)--(M1); \draw[z2] (F1)--(M1); \draw[z2] (F1)--(M2);
  \draw[z2] (M0)--(T);  \draw[z2] (M1)--(T);  \draw[z2] (M2)--(T);
  \draw[z4] (F0) to[bend right=10] (F2);
  \draw[z4] (F2)--(F1);
  \draw[z4] (F0) to[bend right=38] (F1);
  \draw[z4] (M0) to[bend right=10] (M1);
  \draw[z4] (M0) to[bend left=38] (M2);
  \draw[zz] (F0) to[bend left=10] (F2);
  \draw[zz] (M0) to[bend left=10] (M1);
  \draw[zz] (M1)--(M2);
  \draw[zz] (F0) to[out=160,in=180,looseness=1.25] (T);
  \draw[zz] (F1) to[out=20,in=0,looseness=1.25] (T);
  \draw[zz] (F2) to[bend right=30] (T);
  \draw[z4] (F0) to[out=195,in=245,selfloop] (F0);
  \draw[z4] (F2) to[out=250,in=290,selfloop] (F2);
  \draw[z4] (F1) to[out=310,in=360,selfloop] (F1);
  \draw[zz] (F1) to[out=255,in=300,selfloop] (F1);
  \draw[z4] (M0) to[out=150,in=200,selfloop] (M0);
  \draw[zz] (M0) to[out=205,in=255,selfloop] (M0);
  \draw[z4] (M1) to[out=60,in=110,selfloop] (M1);
  \draw[zz] (M1) to[out=125,in=165,selfloop] (M1);
  \begin{scope}[shift={(-3.6,-2.6)}]
    \draw[z2] (0,0)--(0.9,0) node[right,black] {$\mathbb Z_2$};
    \draw[z4] (2.4,0)--(3.3,0) node[right,black] {$\mathbb Z_4$};
    \draw[zz] (4.8,0)--(5.7,0) node[right,black] {$\mathbb Z_2\times\mathbb Z_2$};
  \end{scope}
\end{tikzpicture}
 
\caption{The orbifolding web of table~\ref{tab:web}. Loops are self-dualities under gauging.}
\label{fig:web}
\end{figure}
 
Three features stand out. First, the web is highly connected: the $\mathbb Z_2$ edges alone already connect all seven classes, and $T$ is one $\mathbb Z_2$ or $\mathbb Z_2\times\mathbb Z_2$ gauging away from every other class. Second, the pairs $(F_1,M_0)$, $(F_0,M_2)$ and $(F_2,M_2)$ are not joined by any gauging of order at most four; for them every pair of codes has a quotient of order eight. Third, five of the seven classes are self-dual under some gauging, which we discuss in the next section.
 
\subsection{What gauging does to the spectrum}\label{sec:spectral}
 
To investigate the physical implications of these gaugings, we follow the light states along each connection of the web and record the following properties: the gap of the two theories, how many of the lightest states of $\Lambda(C)$ are charged under $H$, the lightest states of the invariant untwisted sector $\Lambda(I)$ and of the twisted sectors $\Lambda(C'\setminus I)$, and the $\widehat{\su}(2)_1^3$ current factors. Each such factor comes from one of the words $2020$, $0202$ \cite{A3math}. 

Lets us explicitly work out the orbifolding $\mathrm{F}_0\to\mathrm{M}_0$.
Take $C = C_{\mathrm{F}_0} = \langle 0001, 0010\rangle$, whose words are
$(0,0,s,t)$ with $s,t\in\mathbb Z_4$, and
$C' = C_{\mathrm{M}_0} = \langle 0001, 0020, 2000\rangle$, whose words are
$(2a,0,2b,c)$ with $a,b\in\mathbb Z_2$ and $c\in\mathbb Z_4$.
Their intersection $I = \langle 0001, 0020\rangle$ consists of the eight
words $(0,0,2b,c)$, so $H = C'/I \simeq \mathbb Z_2$, generated by
$w = 2000$. For $c=(0,0,s,t)\in C$ one has $\beta(c,w) = 2s$, the symmetry acts as $\chi_w(c) = (-1)^{s}$.

To count states we use the one-variable series of
\cite{A3math}: momentum primaries with labels
$c = (c_1,\dots,c_4)$ contribute $\prod_j \vartheta_{c_j}(t)$, where $t$
records $\Delta$ and
\begin{equation}
  \vartheta_0 = 1 + 12t + \dots, \qquad
  \vartheta_2 = 6t^{1/2} + \dots, \qquad
  \vartheta_1 = \vartheta_3 = 4t^{3/8} + \dots .
\end{equation}
The gaps $3/8$ of $\mathrm{F}_0$ and $\mathrm{M}_0$ are read
off from \cite{A3math}. The states of $\Lambda(C)$ at
$\Delta = 3/8$ come from the words with one odd entry and all others
zero, namely $(0,0,s,0)$ and $(0,0,0,t)$ with $s,t$ odd. Each contributes
$\vartheta_0^3\vartheta_1 = 4t^{3/8} + \dots$, giving $16$ states. Those
with $s$ odd are charged, so $8$ of the $16$ are projected out.

Summing over the words of $I$,
\begin{equation}
  \Theta_{\Lambda(I)} = \vartheta_0^2(\vartheta_0+\vartheta_2)
  (\vartheta_0+\vartheta_2+2\vartheta_1)
  = 1 + 8t^{3/8} + \dots ,
\end{equation}
so the lightest neutral states have $\Delta = 3/8$, with degeneracy $8$.
They are the states with labels $(0,0,0,t)$, $t$ odd. These define the untwisted sector. As for the twisted sector states, the words of $C'\setminus I$ are $(2,0,2b,c)$, and
\begin{equation}
  \Theta_{\Lambda(C'\setminus I)} = \vartheta_2\vartheta_0
  (\vartheta_0+\vartheta_2)(\vartheta_0+\vartheta_2+2\vartheta_1)
  = 6t^{1/2} + \dots ,
\end{equation}
so the lightest twisted states have $\Delta = 1/2$, with degeneracy $6$.
As a check, the two sectors add up to
$\Theta_{\mathrm{M}_0} = \vartheta_0(\vartheta_0+\vartheta_2)^2
(\vartheta_0+\vartheta_2+2\vartheta_1) = 1 + 8t^{3/8} + 18t^{1/2} + \dots$,
in agreement with \cite{A3math}.

And since neither $2020$ nor $0202$ lies in $C$, $\Lambda(C)$ has
no $\widehat{su}(2)_1^3$ factor. The word $2020$ lies in $C'$ but not in
$I$, since its first entry is $2$, while $0202\notin C'$. Hence
$\Lambda(C')$ has one $\widehat{su}(2)_1^3$ factor, and it comes from the
twisted sector: $0\to 1$~(t).
 
\begin{table}[htbp]
\centering
\small
\renewcommand{\arraystretch}{1.18}
\resizebox{\textwidth}{!}{%
\begin{tabular}{@{}llccccl@{}}
\toprule
$C\to C'$ & $H$ & gap: $C\to C'$ & charged at gap of $C$ & lightest untwisted & lightest twisted & $\widehat{\su}(2)_1^3$ factors \\
\midrule
$F_0\to M_0$ & $\mathbb Z_2$ & $3/8\to3/8$ & $8/16$ & $3/8\ (8)$ & $1/2\ (6)$ & $0\to1$ (t) \\
$F_0\to M_1$ & $\mathbb Z_2$ & $3/8\to1/2$ & $16/16$ & $1/2\ (12)$ & $1\ (36)$ & $0\to0$ \\
$F_1\to M_1$ & $\mathbb Z_2$ & $3/4\to1/2$ & $32/64$ & $3/4\ (32)$ & $1/2\ (12)$ & $0\to0$ \\
$F_1\to M_2$ & $\mathbb Z_2$ & $3/4\to1$ & $64/64$ & $1\ (120)$ & $1\ (144)$ & $0\to2$ (t) \\
$F_2\to M_0$ & $\mathbb Z_2$ & $1/2\to3/8$ & $0/12$ & $1/2\ (12)$ & $3/8\ (8)$ & $0\to1$ (t) \\
$F_2\to M_1$ & $\mathbb Z_2$ & $1/2\to1/2$ & $0/12$ & $1/2\ (12)$ & $3/4\ (64)$ & $0\to0$ \\
$M_0\to T$ & $\mathbb Z_2$ & $3/8\to1/2$ & $8/8$ & $1/2\ (18)$ & $1/2\ (6)$ & $1\to2$ (u+t) \\
$M_1\to T$ & $\mathbb Z_2$ & $1/2\to1/2$ & $0/12$ & $1/2\ (12)$ & $1/2\ (12)$ & $0\to2$ (t) \\
$M_2\to T$ & $\mathbb Z_2$ & $1\to1/2$ & $0/264$ & $1\ (264)$ & $1/2\ (24)$ & $2\to2$ (u) \\
\midrule
$F_0\to F_1$ & $\mathbb Z_4$ & $3/8\to3/4$ & $16/16$ & $3/4\ (32)$ & $3/4\ (32)$ & $0\to0$ \\
$F_0\to F_2$ & $\mathbb Z_4$ & $3/8\to1/2$ & $16/16$ & $1/2\ (6)$ & $1/2\ (6)$ & $0\to0$ \\
$F_1\to F_2$ & $\mathbb Z_4$ & $3/4\to1/2$ & $64/64$ & $1\ (84)$ & $1/2\ (12)$ & $0\to0$ \\
$M_0\to M_1$ & $\mathbb Z_4$ & $3/8\to1/2$ & $8/8$ & $1/2\ (6)$ & $1/2\ (6)$ & $1\to0$ \\
$M_0\to M_2$ & $\mathbb Z_4$ & $3/8\to1$ & $8/8$ & $1\ (156)$ & $1\ (108)$ & $1\to2$ (u+t) \\
\midrule
$F_0\to F_2$ & $\mathbb Z_2^2$ & $3/8\to1/2$ & $16/16$ & $1/2\ (12)$ & $7/8\ (96)$ & $0\to0$ \\
$F_0\to T$ & $\mathbb Z_2^2$ & $3/8\to1/2$ & $16/16$ & $1/2\ (12)$ & $1/2\ (12)$ & $0\to2$ (t) \\
$F_1\to T$ & $\mathbb Z_2^2$ & $3/4\to1/2$ & $64/64$ & $1\ (120)$ & $1/2\ (24)$ & $0\to2$ (t) \\
$F_2\to T$ & $\mathbb Z_2^2$ & $1/2\to1/2$ & $0/12$ & $1/2\ (12)$ & $1/2\ (12)$ & $0\to2$ (t) \\
$M_0\to M_1$ & $\mathbb Z_2^2$ & $3/8\to1/2$ & $8/8$ & $1/2\ (12)$ & $3/4\ (64)$ & $1\to0$ \\
$M_1\to M_2$ & $\mathbb Z_2^2$ & $1/2\to1$ & $12/12$ & $1\ (120)$ & $1\ (144)$ & $0\to2$ (t) \\
\bottomrule
\end{tabular}}
\caption{We list spectral effect of each gauging. $\Delta$ denotes dimensions, with the number of momentum primaries in parentheses. The last column counts $\widehat{\su}(2)_1^3$ factors in $\Lambda(C)$ and $\Lambda(C')$; (u) means the factors of $\Lambda(C')$ are untwisted, i.e.\ neutral factors of $\Lambda(C)$, and (t) means they arise in the twisted sector.}
\label{tab:spectral}
\end{table}
 Repeating this computation for the remaining pairs gives table 3, several patterns emerge:
\begin{itemize}
  \item \textbf{Raising the gap:} The theory with the largest gap, $M_2$, is reached from $F_1$, $M_0$, $M_1$ and $T$. In each case all lightest states of the starting theory are charged and projected out, and no twisted state lies below $\Delta=1$.
  \item \textbf{Twisted states set the gap:} Along $F_1\to M_1$, $F_1\to F_2$, $F_1\to T$, $F_2\to M_0$ and $M_2\to T$ the lightest states of $\Lambda(C')$ lie only in the twisted sector, and in all five cases the gap decreases. This is how gauging lowers the gap.
  \item \textbf{Currents:} Starting from a theory without momentum currents ($F_0$, $F_1$, $F_2$ or $M_1$), the $\widehat{\su}(2)_1^6$ of $M_2$ or $T$ arises entirely in the twisted sector. Starting from $M_0$, one $\widehat{\su}(2)_1^3$ factor is neutral and survives, and the other is twisted. Along $M_2\leftrightarrow T$ the current algebra is untouched, as in section~\ref{sec:example}. Conversely, the $\widehat{\su}(2)_1^3$ of $M_0$ is charged and broken along both edges $M_0\to M_1$.
   \item \textbf{Self-dualities:}
Some connections of the web start and end in the same class; these are the
loops of figure~\ref{fig:web}. Along a loop, gauging $H$ maps $\Lambda(C)$
to $\Lambda(C')$ with $C'\neq C$, but $C'$ is obtained from $C$ by a
block-monomial move, which is a T-duality. The orbifold is therefore the
same CFT as the one we started from: $\Lambda(C)$ is self-dual under
gauging $H$. This happens for $\mathrm F_0$, $\mathrm F_1$, $\mathrm F_2$,
$\mathrm M_0$ and $\mathrm M_1$. It cannot happen for $\mathrm M_2$ or
$\mathrm T$, since each of these classes contains a single code. The
resulting non-invertible duality defects are determined in section~\ref{sec:duality}.
\end{itemize}


\section{Surface operators of $\mathcal Z_2$ and non-invertible symmetries}\label{sec:surfaces}
Previously, we described the gaugings from the point of view of the boundary
theories. In this section we ask how they appear in the bulk.
In the SymTFT picture, an operation that maps one topological boundary to another
is the fusion of a topological surface operator of $\mathcal Z_2$ onto the boundary
\cite{KS11,GK21,KOZ22,RSS22}. We classify all surface operators of $\mathcal Z_2$, determine which
of them are invertible, and find the smallest surface that implements each gauging of the web. We
then give the fusion rules of the condensation defects that implement gaugings, and the
non-invertible duality defects that the self-dualities of the web produce on the boundary.

\subsection{Gauging from the bulk}\label{sec:gauging-bulk}

Gauging $H=C'/(C\cap C')$ on the boundary $C$ amounts to condensing the bulk lines of $C'$
on a surface next to it~\cite{KS11,GK21,RSS22}. Figure~\ref{fig:sandwich} shows the three steps.
Inserting the network of $C'$-lines reproduces the twisted partition functions $Z_C[w,w']$ and
their sum~\eqref{eq:orbifold-sum}, and fusing the network onto the boundary replaces $C$ by
$C'$. Each edge of table~\ref{tab:web} is therefore implemented by a surface operator of
$\mathcal Z_2$. The question is which surface, and whether it is invertible \cite{BDR21,BR23,NRS22,ACGHP24}.

\begin{figure}[t]
\centering
\resizebox{\textwidth}{!}{
\begin{tikzpicture}[
  line join=round, line cap=round, font=\small,
  plane/.style={draw=edge, line width=0.5pt},
  wilson/.style={line width=0.9pt},
  endpt/.style={circle, fill, inner sep=1.1pt},
]
\def\dx{1.0}\def\dy{0.7}\def\h{3.0}\def\L{2.6}
 
\newcommand{\slab}[1]{%
  \begin{scope}[shift={(#1,0)}]
    \fill[bulk] (0,0) -- (\L,0) -- (\L,\h) -- (0,\h) -- cycle;
    \fill[bulk!85!black!8] (0,\h) -- (\L,\h) -- (\L+\dx,\h+\dy) -- (\dx,\h+\dy) -- cycle;
    \draw[edge, line width=0.4pt] (0,0) -- (\L,0);
    \draw[edge, line width=0.4pt] (\dx,\h+\dy) -- (\L+\dx,\h+\dy);
    \draw[edge, line width=0.4pt] (0,\h) -- (\L,\h);
    \draw[edge, densely dotted, line width=0.4pt] (\dx,\dy) -- (\L+\dx,\dy);
    \draw[plane, fill=white, fill opacity=0.85,
          decorate, decoration={snake, amplitude=0.6pt, segment length=6pt}]
      (0,0) -- (\dx,\dy) -- (\dx,\dy+\h) -- (0,\h) -- cycle;
    \node[anchor=north] at (0.5*\dx,-0.05) {$\Psi_c$};
    \node at (0.38*\L+0.5*\dx,0.5*\h+0.75*\dy) {$\mathcal Z_2$};
  \end{scope}}
 
\newcommand{\tbc}[3]{%
  \begin{scope}[shift={(#1,0)}]
    \fill[#2!12] (\L,0) -- (\L+\dx,\dy) -- (\L+\dx,\dy+\h) -- (\L,\h) -- cycle;
    \draw[#2, line width=1.4pt] (\L,0) -- (\L+\dx,\dy) -- (\L+\dx,\dy+\h) -- (\L,\h) -- cycle;
    \node[#2, anchor=north] at (\L+0.5*\dx,-0.05) {#3};
  \end{scope}}
 
\slab{0}
\tbc{0}{bndC}{$C$}
\begin{scope}
  \draw[wilson, bndC] (0.5*\dx,1.15) -- (\L+0.5*\dx,1.15) node[endpt]{};
  \node[bndC, anchor=south] at (1.55,1.15) {$c\in C$};
\end{scope}
\node at (1.8,-0.95) {$Z_C$};
\node[font=\normalsize] at (-0.35,\h+0.6) {(a)};
 
\def\ob{5.4}
\slab{\ob}
\tbc{\ob}{bndC}{$C$}
\begin{scope}[shift={(\ob,0)}]
  \def\xm{\L-0.42}
  \fill[bndCp!6] (\xm,0) -- (\xm+\dx,\dy) -- (\xm+\dx,\dy+\h) -- (\xm,\h) -- cycle;
  \draw[bndCp!60, densely dashed, line width=0.4pt]
    (\xm,0) -- (\xm+\dx,\dy) -- (\xm+\dx,\dy+\h) -- (\xm,\h) -- cycle;
  \foreach \t in {0.2,0.4,0.6,0.8}{
    \draw[bndCp, line width=0.6pt] (\xm+\t*\dx,\t*\dy) -- (\xm+\t*\dx,\t*\dy+\h);
  }
  \foreach \s in {0.12,0.29,0.46,0.63,0.80}{
    \draw[bndCp, line width=0.6pt] (\xm,\s*\h) -- (\xm+\dx,\dy+\s*\h);
  }
  \foreach \t in {0.2,0.4,0.6,0.8}{
    \foreach \s in {0.12,0.29,0.46,0.63,0.80}{
      \fill[bndCp] (\xm+\t*\dx,\t*\dy+\s*\h) circle (0.9pt);
    }
  }
  \draw[wilson, bndC] (0.5*\dx,0.55) -- (\L+0.5*\dx,0.55) node[endpt]{};
  \node[bndCp, anchor=south, align=center, font=\footnotesize]
    at (\xm+0.5*\dx,\h+\dy+0.08) {$\dfrac{1}{|H|}\!\displaystyle\sum_{w\in H}$};
\end{scope}
\node[bndCp, align=left, font=\footnotesize, anchor=west] at (\ob+\L+\dx+0.15,\h-0.25)
  {lines of $C'$\\ $H=C'/(C\cap C')$};
\node at (\ob+1.8,-0.95) {$\dfrac{1}{|H|}\displaystyle\sum_{w,w'\in H} Z_C[w,w']$};
\node[font=\normalsize] at (\ob-0.35,\h+0.6) {(b)};
 
\def\oc{11.6}
\slab{\oc}
\tbc{\oc}{bndCp}{$C'$}
\begin{scope}[shift={(\oc,0)}]
  \draw[wilson, bndCp] (0.5*\dx,1.15) -- (\L+0.5*\dx,1.15) node[endpt]{};
  \node[bndCp, anchor=south] at (1.55,1.15) {$c'\in C'$};
\end{scope}
\node at (\oc+1.8,-0.95) {$Z_{C'}$};
\node[font=\normalsize] at (\oc-0.35,\h+0.6) {(c)};
 
\draw[-{Stealth[length=5pt]}, line width=0.7pt] (3.95,1.9) -- (5.1,1.9)
  node[midway, above, align=center, font=\footnotesize] {insert\\network};
\draw[-{Stealth[length=5pt]}, line width=0.7pt] (\ob+\L+\dx+0.35,1.6) -- (\oc-0.3,1.6)
  node[midway, below, align=center, font=\footnotesize] {fuse onto\\boundary};
\end{tikzpicture}}
\caption{Gauging between boundaries, as seen from the bulk. (a) The slab $\mathcal Z_2\times[0,1]$
with the relative theory $\Psi$ on the left and the topological boundary $C$ on the right. Bulk
lines with labels in $C$ end on $C$, and shrinking the slab gives $Z_C$. (b) A network of bulk
lines with labels in $C'$ is inserted on a surface parallel to $C$, averaged over
$H=C'/(C\cap C')$. Evaluating it gives the orbifold sum~\eqref{eq:orbifold-sum}. (c) Fusing the
network onto the boundary turns $C$ into $C'$, and the slab gives $Z_{C'}$.}
\label{fig:sandwich}
\end{figure}
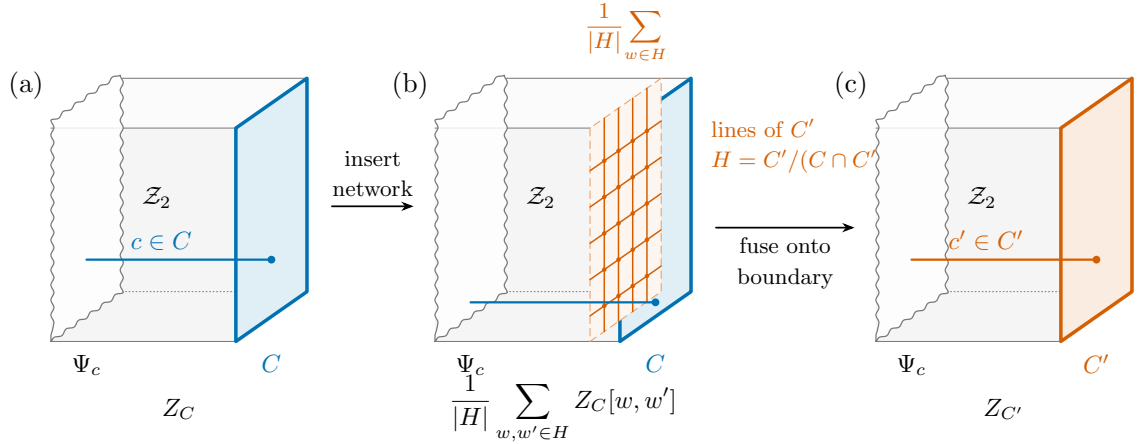

\subsection{Classification of the surfaces}\label{sec:surf-class}

The topological surface operators of an abelian Chern-Simons theory with line group $M$ are
labelled by the Lagrangian subgroups $L\subset M\oplus\overline M$, where $\overline M$ carries
the form $-\mathsf q$ \cite{KS-surf,FSV}. This is the folding trick where a surface in $\mathcal Z_d$
is a topological boundary of $\mathcal Z_d\times\overline{\mathcal Z_d}$. A pair $(x,y)\in L$ means
that the line $x$ on one side of the surface can pass through it and emerge as the line $y$ on the
other side. The map
\[
  \big((a,b),(a',b')\big)\longmapsto\big((a,a'),(b,-b')\big)
\]
identifies $(M_d\oplus\overline M_d,\mathsf q\oplus-\mathsf q)$ with $(M_{2d},\mathsf q)$, since
$(a,a')\cdot(b,-b')=a\cdot b-a'\cdot b'$. Hence, the topological surface operators of $\mathcal Z_d$ are in bijection with the even self-dual
$A_3$ gluing codes with $2d$ hyperbolic pairs. 

The count follows from the reduction modulo two of section~\ref{sec:web-result}. A Lagrangian
subgroup of $M_{n}$ whose reduction is a totally singular subspace of dimension $k$ in $\F_2^{2n}$
has $2^{\binom k2}$ lifts (appendix~\ref{app:surfaces}). The number of totally singular
$k$-subspaces of the split quadratic space $\F_2^{2n}$ is
$\left[{n\atop k}\right]_2\prod_{i=n-k}^{n-1}(2^i+1)$, so
\begin{equation}\label{eq:Nd}
\begin{gathered}
  N_n\equiv|\mathcal L_n|=\sum_{k=0}^n\left[{n\atop k}\right]_2\,2^{\binom k2}\prod_{i=n-k}^{n-1}(2^i+1),\\
  N_1=3,\qquad N_2=22,\qquad N_3=486,\qquad N_4=36766 .
\end{gathered}
\end{equation}
These numbers agree with the counts of $\Z_4$ codes in \cite{Dymarsky,BDS25}. 

A surface $S_L$ is described by what it does to the lines that hit it. A pair $(x,y)\in L$ means
that the line $x$ on the left of the surface can continue as the line $y$ on the right. The lines
that end on the surface, from the left and from the right, are
\begin{equation}\label{eq:kernels}
  K_1=\{x:\ (x,0)\in L\},\qquad K_2=\{y:\ (0,y)\in L\}.
\end{equation}
Physically, $K_1$ and $K_2$ are the anyons condensed on the two faces of the surface. A line can
end on the surface only if it is a boson with trivial braiding with the other condensed lines, so
$K_1$ and $K_2$ are isotropic. A line that braids nontrivially with $K_1$ cannot reach the
surface from the left: it is confined. The lines that can reach it are those of $K_1^\perp$, and
two of them that differ by an element of $K_1$ continue in the same way. The lines that genuinely
cross are therefore labeled by $K_1^\perp/K_1$ on the left and $K_2^\perp/K_2$ on the right, and
the surface matches them through an isometry
\begin{equation}\label{eq:phi}
  \varphi:\ K_1^\perp/K_1\longrightarrow K_2^\perp/K_2,
\end{equation}
whose graph is $L/(K_1\oplus K_2)$ \cite{FSV}.
The left projection of $L$ is $K_1^\perp$, and each of its elements is paired with a coset of
$K_2$ on the right, so $|L|=|K_1^\perp|\,|K_2|$. Since $|L|=|M_2|=|K_1^\perp|\,|K_1|$, this gives
$|K_1|=|K_2|$. We call this common value the
quantum dimension of $S_L$ and denote it by $\kappa(S_L)=|K_1|=|K_2|$. It is the number of lines the surface condenses on each face. For a
condensation defect $S_A$, with $K_1=K_2=A$, it reproduces the fusion rule $S_A\times S_A=|A|\,S_A$
of \cite{RSS22}.

The surface is invertible exactly when it condenses nothing, $K_1=K_2=0$. Every line then
crosses, $\varphi$ is a symmetry of the whole line group, and $L$ is the graph of an element of
$\Ot(M_2,\mathsf q)$. To illustrate, let us consider the following case:

\paragraph{One pair case.} For a single $\Z_4$ gauge theory, $d=1$, with lines $(a,b)\in\Z_4^2$ and spin
$ab/4$, we have $N_2=22$ surfaces that distribute as:
\begin{itemize}
  \item Quantum dimension one: the four isometries of $(\Z_4^2,ab/4)$, namely the identity, charge
  conjugation $(a,b)\mapsto(-a,-b)$, electric--magnetic duality $(a,b)\mapsto(b,a)$ and their
  product.
  \item Quantum dimension two: $K_1$ and $K_2$ are subgroups of order two, generated by $(2,0)$,
  $(0,2)$ or $(2,2)$. If $K_i$ is $\langle(2,0)\rangle$ or $\langle(0,2)\rangle$, the crossing lines
  form $\Z_2\times\Z_2$ with spins $0,0,\tfrac12$, which has two isometries. This gives
  $2\times2\times2=8$ surfaces. If $K_1=K_2=\langle(2,2)\rangle$, the crossing lines have spins
  $\tfrac14,\tfrac34$ and admit a single isometry, which gives one surface. The two kinds of quotient
  are not isometric, so they cannot be mixed. In total there are nine surfaces.
  
  \item Quantum dimension four: $K_1$ and $K_2$ are Lagrangian. No line crosses, and the surface
  is a pair of topological boundaries placed back to back. There are $3\times3=9$ of them.
\end{itemize}

\paragraph{Two pairs case.} For $\mathcal Z_2$ the total number of surfaces is $N_4=36766$, by the
counting formula~\eqref{eq:Nd}. Two families can be identified directly:

The group $\Ot(M_2,\mathsf q)$ is generated by the block-monomial moves
$\sigma$, $\varepsilon_i$, $\nu_i$ and the $B$-field shift $h_1$ of~\eqref{eq:B-shift}.
To see this, write $q_4(a,b)=a\cdot b\pmod 4$ for the integer-valued form of~\eqref{eq:beta-q},
and distinguish it from the spin form $\mathsf q=\tfrac34q_4\pmod 1$; the two have the same
isometry group. Reduction modulo two maps $\Ot(M_2,\mathsf q)$ to the orthogonal group $\Ot^+(4,2)$ of
$(\F_2^4,\bar q)$, which has order $72$. The reductions of $\sigma$, $\varepsilon_i$, $\nu_i$ and
$h_1$ generate all of $\Ot^+(4,2)$, so the map is onto. Its kernel consists of the isometries
$g=1+2X$. Since
\[
  q_4(x+2Xx)=q_4(x)+2\beta(x,Xx)\pmod 4,
\]
$g$ preserves $q_4$
exactly when $\bar\beta(x,Xx)=0$ for all $x$, that is, when $(x,y)\mapsto\bar\beta(x,Xy)$ is
alternating. Alternating forms on $\F_2^4$ form a space of dimension $\binom42=6$, so the kernel
is elementary abelian of order $2^6=64$, and $|\Ot(M_2,\mathsf q)|=72\cdot64=4608$.
Surjectivity modulo two alone does not show that the named moves generate the whole group; one
also needs the kernel. The six elements
\[
  \nu_1,\quad \nu_2,\quad h_1^2,\quad
  \varepsilon_1h_1^2\varepsilon_1,\quad
  \varepsilon_2h_1^2\varepsilon_2,\quad
  (\varepsilon_1\varepsilon_2)\,h_1^2\,(\varepsilon_1\varepsilon_2)
\]
lie in the kernel and are independent in it, so they generate it. Together with the generators
of the reduction, this proves that $\sigma$, $\varepsilon_i$, $\nu_i$ and $h_1$ generate all
$4608$ elements of $\Ot(M_2,\mathsf q)$.
These are the invertible surfaces, they form the set of T-dualities and fractional $B$-field shifts of
section~\ref{sec:bulk-sym}. On the other hand, a surface with $K_1$ and $K_2$ Lagrangian lets
no line cross. It is the factorized surface
\begin{equation}\label{eq:factorized}
  S_L=|C\rangle\langle C'|,\qquad C=K_1,\ C'=K_2,
\end{equation}
which cuts the bulk into two slabs, one ending on $C$ and the other on $C'$. There are
$22^2=484$ of them. The remaining $36766-4608-484=31674$ surfaces are non-invertible and condense
a proper, non-Lagrangian isotropic subgroup on each face, of order two, four or eight. They
include the condensation defects that implement the gaugings of section~\ref{sec:web}.

\subsection{Condensation defects}\label{sec:cond}

The surfaces that implement gaugings in figure~\ref{fig:sandwich} are condensation defects
\cite{RSS22}. The idea is simple. Take a subgroup $A$ of bulk lines that are mutually local
bosons, that is, an isotropic subgroup $A\subset M_2$. Gauge these lines on a surface $\Sigma$ in
the bulk, by summing over networks of $A$-lines on $\Sigma$. The result is a topological surface
$S_A$. On $\Sigma$ the lines of $A$ are condensed: they can end there, so $K_1=K_2=A$.

A line $x$ that braids
nontrivially with some line of $A$ cannot cross the surface, because its phase around the
condensate is not single valued. It is confined. A line that braids trivially with all of $A$,
that is $x\in A^\perp$, crosses freely. As it crosses, it can absorb or emit lines of the
condensate, so it emerges as $x+a$ for some $a\in A$. The surface is therefore described by
\begin{equation}\label{eq:LA}
  L_A=\{(x,x+a):\ x\in A^\perp,\ a\in A\}.
\end{equation}
Its quantum dimension is $|A|$, the number of condensed lines.

Now, we push $S_A$ onto a topological boundary $C$, two things happen to
the condensate of the boundary. The lines of $C$ that braid nontrivially with $A$ are confined by
the surface and disappear from the boundary. The lines of $A$, which can now end on the surface,
join the boundary condensate. The new boundary is
\begin{equation}\label{eq:SA-on-C}
  L_A\circ C=(C\cap A^\perp)+A .
\end{equation}
Indeed, $x\in L_A\circ C$ exactly when $x\in A^\perp$ and $x+a\in C$ for some $a\in A$, that is,
$L_A\circ C=A^\perp\cap(C+A)$, which equals $(C\cap A^\perp)+A$ since $A\subset A^\perp$.

On the boundary theory this is a gauging. The confined lines of $C$ correspond to the states of
$\Lambda(C)$ that are charged under the gauged symmetry and are projected out. The added lines of
$A$ correspond to the twisted sectors that appear after gauging. For example, the gauging from
$F_0$ to $M_0$ of section~\ref{sec:spectral} is the condensation of the single line $A=\langle
2000\rangle$. With $C=\{(0,0,s,t)\}$, the condensed line braids with $(0,0,s,t)$ by $(-1)^s$, so the
lines with $s$ odd are confined: these are the eight charged states at $\Delta=3/8$ that are
projected out. The line $2000$ joins the boundary, which becomes $\{(2a,0,2b,t)\}=C_{M_0}$. Its
lines with $a=1$ are the twisted sector, which starts at $\Delta=1/2$.

More generally, considering $C,C'$ even and self-dual, $I=C\cap C'$, and let $A\subset C'$ be a subgroup with $A+I=C'$. Then $L_A\circ C=C'$. In particular the gauging of $H=C'/I$ on the boundary $C$ is implemented by a
condensation defect of quantum dimension $|A|\geq|H|$. $I$ is the part of the condensate that survives the gauging, and $A$ supplies the new
condensed lines. The smallest choice is $|A|=|H|$, when $A$ contains only the new lines. If the new
lines cannot be separated from $I$, the surface must condense more than $|H|$ lines.

Two condensation defects can be brought together. If the condensates are
mutually local, $\beta(A,B)=0$, the combined surface simply condenses both: lines that are confined
by either surface are confined, and the crossing lines can absorb lines of either condensate.
\label{prop:fusion}
Let $A,B\subset M_2$ be isotropic with $\beta(A,B)=0$. Then $A+B$ is isotropic and
\begin{equation}\label{eq:fusion}
  S_A\times S_B=\mathcal Z(A\cap B)\,S_{A+B},
\end{equation}
where $\mathcal Z(G)$ is a decoupled two-dimensional $G$ gauge theory, with $|G|$ vacua, as in
\cite{RSS22}. In particular $S_A\times S_A=\mathcal Z(A)\,S_A$, and $S_A\times S_B=S_{A+B}$ when
$A\cap B=0$.
A pair $(x,z)$ lies in $L_B\circ L_A$ when $x\in A^\perp$, $y=x+a\in B^\perp$ and $z=y+b$ for some
$a\in A$, $b\in B$. Since $A\subset B^\perp$, the condition $y\in B^\perp$ is $x\in B^\perp$, so
$x\in(A+B)^\perp$ and $z-x\in A+B$. This is $L_{A+B}$. Evenness of $A+B$ follows from
$\mathsf q(a+b)=\mathsf q(a)+\mathsf q(b)+\tfrac34\beta(a,b)$ and $\beta(A,B)=0$.

When the two surfaces are brought together,
the lines condensed on both of them, those of $A\cap B$, can be created on one surface and
annihilated on the other. The region between the two surfaces then carries a two-dimensional
gauge theory for $A\cap B$, which decouples. 

On the boundary, gauging by $S_A$ and then by
$S_{A'}$, with $\beta(A,A')=0$, is the same as a single gauging by $S_{A+A'}$, up to the decoupled
factor $\mathcal Z(A\cap A')$. In particular, gauging the same symmetry twice, $S_A\times S_A$, is again $S_A$, accompanied by a decoupled TQFT.

\subsection{Which surfaces implement which gaugings}\label{sec:surf-action}

A topological boundary $C$ is a condensate, meaning that the lines of $C$ can end on it. Now push a surface
$S_L$ onto this boundary. A line $x$ coming from the bulk first meets the surface. If it crosses,
it continues as a line $y$ with $(x,y)\in L$, and it can end on the boundary if $y\in C$. If it is
condensed on the surface, it ends there directly. The result is a new boundary, whose condensed
lines are all the lines that can now end:
\begin{equation}\label{eq:L-on-C}
  L\circ C=\{x\in M_2:\ (x,y)\in L\ \text{for some}\ y\in C\}.
\end{equation}
Every way of turning one boundary into another, and in particular every gauging of
section~\ref{sec:web}, is of this form. We distinguish the following:

\begin{itemize}
    \item \textbf{Invertible surfaces:} An invertible surface condenses nothing. It only relabels the
lines, $x\mapsto gx$ with $g\in\Ot(M_2,\mathsf q)$, and the new boundary is $gC$. A relabeling
preserves the order of every line, so it cannot turn a condensed line of order four into a line of
order two. The number $j(C)$ of independent $\Z_4$ charges condensed on the boundary, that is the
number of $\Z_4$ factors of $C$, is therefore unchanged. This is the physical content of the three
orbits of section~\ref{sec:bulk-sym}: free codes have $j=2$, mixed codes $j=1$ and the torsion code
$j=0$, and invertible surfaces act within each orbit. 
\item \textbf{Non-invertible surfaces:} A surface that condenses lines can change $j$. In the
$F_0\to M_0$ example of section~\ref{sec:cond}, the condensed line $2000$ confines the lines of $C$
with $s$ odd, and one independent $\Z_4$ charge is replaced by a $\Z_2$ charge: $j$ drops from two
to one. Each condensed line of order two can confine at most one independent $\Z_4$ charge of the
boundary.
\end{itemize}

The precise statement is the following proposition. It concerns the quantum
dimension $\kappa(S_L)=|K_1|=|K_2|$ of section~\ref{sec:surf-class} and the action~\eqref{eq:L-on-C}
of surfaces on boundary codes.

\begin{proposition}\label{prop:min-dim}
For Lagrangian codes $C,C'\subset M_2$, let $j(C)$ be the number of $\Z_4$ factors of $C$. Then
\begin{equation}\label{eq:min-dim}
  \min_{L\circ C=C'}\kappa(S_L)=2^{|j(C)-j(C')|}.
\end{equation}
It is $1$ within an orbit, $2$ between the free and the mixed orbits and between the mixed and
the torsion orbits, and $4$ between the free and the torsion orbits.
\end{proposition}

\begin{proof}
An invertible surface has $\kappa=1$ and preserves the group type. Conversely, $\Ot(M_2,\mathsf q)$
is transitive on each of the free, mixed and torsion orbits. Thus the minimum is one when
$j(C)=j(C')$, and is at least two otherwise.

We next exclude $\kappa=2$ between a free code and the torsion code. For an order-two isotropic
subgroup $K$, the reduced image of a boundary $C$ is
\[
  C_K=\frac{(C\cap K^\perp)+K}{K}\simeq\frac{C\cap K^\perp}{C\cap K}.
\]
If $C\simeq\Z_4^2$, this group is always $\Z_4\oplus\Z_2$. Indeed, if $K\subset C$, it is $C/K$.
Otherwise $C\cap K=0$, and $C\cap K^\perp$ is the kernel of a nontrivial character
$C\to\{\pm1\}$, hence an index-two subgroup of $\Z_4^2$. Both cases give $\Z_4\oplus\Z_2$.

For $T=2M_2$, every order-two $K$ lies in $T$, and $T\subset K^\perp$. Thus $T_K=T/K\simeq\Z_2^3$.
A surface with $\kappa=2$ induces the isometry $K_1^\perp/K_1\simeq K_2^\perp/K_2$
of~\eqref{eq:phi}, which would have to identify these two reduced boundary groups. Their exponents
differ, so this is impossible in either direction. The free--torsion minimum is at least four.

To attain all the bounds, consider one hyperbolic pair $M_1=\Z_4e\oplus\Z_4f$, with
$q_4(ae+bf)=ab\pmod4$, and put
\[
  E=\Z_4e,\qquad T_1=2M_1,\qquad K=\langle2e\rangle.
\]
The quotient $K^\perp/K$ is the hyperbolic quadratic module on $\F_2^2$, with basis $e+K$,
$2f+K$. Exchange these two basis vectors and take the inverse image of its graph in
$K^\perp\oplus K^\perp$. This defines a Lagrangian correspondence $R$ with $K_1=K_2=K$, and
therefore $\kappa(S_R)=2$. Explicitly,
\[
  R=\left\{\bigl((a,2b),(a',2b')\bigr):
  \begin{array}{l}
  a,a'\in\Z_4,\ b,b'\in\F_2,\\
  a\bmod2=b',\quad a'\bmod2=b
  \end{array}\right\}.
\]
It satisfies $R\circ E=T_1$ and $R\circ T_1=E$. The standard free, mixed and torsion boundaries
of $M_2$ are $E\oplus E$, $E\oplus T_1$ and $T_1\oplus T_1$. Applying $R$ to one pair gives
surfaces of dimension two between adjacent types; applying it to both pairs gives dimension four
between free and torsion. Composing on the two sides with invertible surfaces reaches arbitrary
codes in the corresponding orbits without changing $\kappa$.
\end{proof}

The proposition has three physical consequences.
\begin{itemize}
\item \textbf{Same group type, invertible bulk surface:} if the source and target codes have
the same group type, an invertible bulk surface $g$ with $gC=C'$ relates their boundary
conditions. If $g$ is a T-duality, the two theories coincide, and we obtain the self-dualities.
Otherwise $g$ involves a fractional $B$-field shift. This does not mean that the gauging
interface itself is invertible. For instance, the self-dual $\Z_4$ gaugings of
section~\ref{sec:duality} have boundary duality lines of dimension $\sqrt{|H|}=2$, although an
invertible bulk surface relates the two codes. Proposition~\ref{prop:min-dim} is a statement
about the minimum of $\kappa$, not about these boundary-line dimensions.
  \item \textbf{Different group types, non-invertible bulk surfaces:} if the group types differ,
  every bulk surface relating the two boundary conditions is non-invertible. Changing the number
  of condensed $\Z_4$ charges requires confining some lines, which loses information.
  For the $\Z_2$ gaugings between free and mixed codes and between mixed codes and $T$, and for the
  $\Z_2\times\Z_2$ gaugings between free codes and $T$, the gauge group already has the minimal
  size $2^{|j-j'|}$.
  \item \textbf{Order-eight pairs are close in the bulk.} The pairs $(F_1,M_0)$, $(F_0,M_2)$ and
  $(F_2,M_2)$ are not related by any gauging of order at most four. In the bulk they are related by
  a surface with $\kappa=2$: a $\Z_2$ condensation, which changes $j$ by one, followed by
  an invertible surface, which moves within the orbit. On the boundary this is a $\Z_2$ gauging
  followed by a fractional $B$-field shift. The order of the gauge group is therefore not a bulk
  invariant, while the minimal $\kappa$ depends only on how many $\Z_4$ charges must be confined.
\end{itemize}

\subsection{Duality defects on the boundary}\label{sec:duality}

Five classes of the web are self-dual under gauging. In each case there is a block-monomial move
$g$, which is a T-duality, such that $C'=gC$ and $C'$ is obtained from $C$ by gauging
$H=C'/(C\cap C')$. Composing the gauging interface with the T-duality gives a topological line $D$
of the theory $\Lambda(C)$, from the theory to itself. This line is non-invertible. It obeys
\begin{equation}\label{eq:TY-fusion}
  D\times\overline D=\bigoplus_{h\in H}\eta_h,\qquad D\times\eta_h=\eta_h\times D=D ,
\end{equation}
where $\eta_h$ are the lines of the gauged symmetry $H$. In all eight self-dualities of
table~\ref{tab:pairs}, $g$ is an involution, and we take $D=\overline D$ and the lines $\{\eta_h,D\}$ form a
Tambara-Yamagami category $\mathrm{TY}(H,\chi,\epsilon)$ \cite{TY98}. The symmetric bicharacter
is determined by the code. For $h,h'\in C'$ representing elements of $H$,
\begin{equation}\label{eq:chi}
  \chi(h,h')=\exp\!\Big(2\pi i\,\tfrac34\,\beta(gh,h')\Big).
\end{equation}
This is well defined because $g(C+C')=C+C'$ in all eight cases (appendix~\ref{app:surfaces}). We
find:
\begin{equation}\label{eq:TY-table}
\begin{array}{llll}
  \toprule
  \text{class} & H & g & \chi\\
  \midrule
  F_0 & \Z_4 & \varepsilon_1 & \chi(k,l)=i^{-kl}\\
  F_1 & \Z_4 & \varepsilon_1 & \chi(k,l)=i^{-kl}\\
  F_2 & \Z_4 & \varepsilon_1 & \chi(k,l)=i^{-kl}\\
  M_0 & \Z_4 & \varepsilon_2 & \chi(k,l)=i^{-kl}\\
  M_1 & \Z_4 & \varepsilon_1 & \chi(k,l)=i^{-kl}\\
  F_1 & \Z_2\times\Z_2 & \nu_1 & \text{hyperbolic: } \chi(h,h)=1\ \text{for all}\ h\\
  M_0 & \Z_2\times\Z_2 & \sigma & \text{hyperbolic: } \chi(h,h)=1\ \text{for all}\ h\\
  M_1 & \Z_2\times\Z_2 & \varepsilon_1\varepsilon_2 & \text{diagonal: } \chi(h,h')=(-1)^{h_1h_1'+h_2h_2'}\\
  \bottomrule
\end{array}
\end{equation}

For $\Z_4$, $k,l$ are the components along a generator of order four. For $M_1$, $h_1,h_2$ are the
components in the basis $1100+I$, $1300+I$ of $H$. For $\Z_2\times\Z_2$, the hyperbolic and the
diagonal bicharacters are the two inequivalent nondegenerate symmetric bicharacters. 

As an example, consider $F_0$ with $g=\varepsilon_1$. Here $C=\{(0,0,s,t)\}$,
$C'=\{(s,0,0,t)\}$, $I=\{(0,0,0,t)\}$, and $H=C'/I\simeq\Z_4$ is generated by $h=1000$. The
exchange $a_1\leftrightarrow b_1$ preserves $C+C'=\{(s,0,s',t)\}$. Since $gh=0010$ and
$\beta(0010,1000)=1$, we find $\chi(h,h)=e^{2\pi i\cdot3/4}=i^{-1}$, and by bilinearity
$\chi(k,l)=i^{-kl}$. The other seven cases are checked in the same way.

In the bulk, these non-invertible lines have an invertible origin. A duality defect of a
boundary theory is the end line of an invertible bulk surface that preserves the relative theory
$\Psi$ but not the topological boundary \cite{KOZ22}. Here that surface is the T-duality $g$:
it preserves $\Psi$, and $gC=C'\neq C$. The theories $M_2$ and $T$ have no such defects, since
each of these classes contains a single code and is fixed by every T-duality.

\section{TQFT gravity for the $A_3$ family}
\label{sec:ensemble}
So far we have fixed the bulk and changed the boundary. In this section we do the opposite: we fix
the boundary torus and sum over bulk geometries. The abelian
Chern-Simons theory $\mathcal Z_d$ plays the role of the gravitational theory, and its sum over
topologies plays the role of the gravitational path integral, as in the $U(1)$ gravity dual to the
average over Narain moduli space \cite{MW20,AJCHT21}.

Dymarsky and Shapere showed what this sum must be \cite{Dymarsky}. It runs
over all three-manifolds, not only smooth handlebodies. Its result is the uniform average over all
topological boundary conditions, that is, over all code CFTs. For prime or square-free $N$, the
handlebodies alone reproduce this average. For $\Z_4$ they do not. At genus one, one additional
geometry is needed: a solid torus with a line defect along its core, which we call the torsion
topology. Its weight is fixed by the requirement that the bulk sum reproduce the uniform average.
That the boundary ensemble must be uniform was derived from first principles, by a genus-reduction
argument, in \cite{Dymarsky}. For $N=p^2$, \cite[appendix~D.3]{BDS25} obtains the genus-one weights
by matching the code sum, and expressly leaves a direct evaluation of the topology sum for future work.

We apply this to the $A_3$ family, where each boundary condition is a concrete CFT. This adds four
things. First, we give a direct proof of the handlebody sum for any number of pairs. It takes the
form of a finite Siegel-Weil identity: the sum over the six handlebodies equals an average over
the free codes, corrected by a multiple of the torsion state. This identity agrees with $p=2$ of
\cite[appendix~D.3]{BDS25}. Our proof is based on computing the six handlebody states in the
Weil representation, which shows which handlebodies produce each term. Second, we give the weights
of the handlebodies and of the torsion topology in closed form for any number of pairs. Third, we
identify the two bulk contributions with theories of the family. For two pairs, the handlebodies
produce the average of the free theories $F_0,F_1,F_2$, and the torsion topology produces the
torsion theory $T$. Fourth, we show explicitly why the handlebodies alone are not enough. For three
or more pairs, the handlebody state cannot be written as a nonnegative combination of the
boundary vectors, and the torsion topology is what restores positivity. This concerns
the bulk state; it does not by itself rule out a positive decomposition of the scalar partition
function.

\subsection{Handlebodies and the torsion topology}\label{sec:ens-handlebody}

We write the partition vector of the
relative theory as $\Psi_c=\Theta_c/|\eta|^{6d}$, with $c\in M_d$, and let $\{e_c\}$ be the standard
basis of $\C[M_d]$. With $Q(c)=\mathsf q(c)$ and $B(c,c')=\tfrac34\,\beta(c,c')$ modulo one,
the modular group acts by the Weil representation~\cite{weil}
\begin{equation}\label{eq:weil}
  \rho(T)\,e_c=e^{2\pi i Q(c)}e_c,\qquad
  \rho(S)\,e_c=4^{-d}\sum_{c'\in M_d}e^{-2\pi i B(c,c')}e_{c'} .
\end{equation}
Since $4Q$ is integer valued, $\rho$ factors through
$\mathrm{SL}(2,\Z_4)$. Since $K$ has signature $(3d,3d)$, $\rho(S^2)=\rho(-1)$ acts as $e_c\mapsto e_{-c}$,
without a phase. For $f=\sum_cf(c)e_c$, the function $\Theta[f]/|\eta|^{6d}$ of~\eqref{eq:Theta-f}
is modular invariant whenever $f$ is $\rho$-invariant. This follows from the transformation laws of
the block theta functions used in the proof of proposition~\ref{prop:twisted}.

The path integral of $\mathcal Z_d$ on a solid torus prepares a state on its boundary torus.
With no insertion it prepares the vacuum $e_0$. A modular transformation $\gamma$ changes which
cycle of the boundary torus is contractible, and the corresponding solid torus prepares
$\rho(\gamma)e_0$. The vector $e_0$ is fixed by $\Gamma_\infty=\langle T,-1\rangle$, and the
principal congruence subgroup $\Gamma(4)$ acts trivially through $\rho$. Up to an overall
multiplicity, the handlebody contribution is
\begin{equation}\label{eq:P-e0}
  \mathcal P_d[v]=\sum_{\gamma\in\mathrm{SL}(2,\Z_4)/\Gamma_\infty}\rho(\gamma)\,v ,
  \qquad
  Z_{\rm grav}=\big\langle\Psi,\mathcal P_d[e_0]\big\rangle .
\end{equation}
The sum has six terms, since $|\mathrm{SL}(2,\Z_4)|=48$ and the image of $\Gamma_\infty$ has order
eight. It is the finite analogue of the real-analytic Eisenstein series of the
Maloney-Witten sum~\cite{MW20,AJCHT21}.

The torsion topology of \cite[equation~(2.20)]{Dymarsky} prepares the state
\begin{equation}\label{eq:torsion-state}
  2^{-d}\sum_{a,b\in\Z_2^d}e_{(2a,2b)}=2^{-d}\,v_{T_d},
\end{equation}
where $T_d=2M_d$ is the torsion code and $v_C=\sum_{c\in C}e_c$ denotes the code vector.
Geometrically, it is a solid torus with a line defect along its core: the core insertion is the
normalized sum of the Wilson lines with labels in $T_d=2M_d$. Around the cycle that would be
contractible in a handlebody, only the Wilson lines with labels in $2M_d$ are trivial
\cite{Dymarsky}. Since $T_d$ is Lagrangian, the state is proportional to the
boundary state of the topological boundary $T_d$, and $\langle\Psi,v_{T_d}\rangle=Z_{T_d}$. In the
$A_3$ family, the torsion topology therefore contributes the partition function of the torsion
theory: the theory with positive lattice $\Z^{6d}$, gap $1/2$ and current algebra
$\widehat{\su}(2)_1^{3d}$. The vector $v_{T_d}$ is $\rho$-invariant, so $\mathcal P_d[v_{T_d}]=6v_{T_d}$.

\subsection{The handlebody contribution}\label{sec:ens-SW}

Call a Lagrangian subgroup $C\subset M_d$ free when $C\simeq\Z_4^d$, and let $\mathcal F_d$ be the set
of free Lagrangians. For $d=2$, $\mathcal F_2$ is the free orbit $\{F_0,F_1,F_2\}$ of
section~\ref{sec:bulk-sym}, and $T_2$ is the code $T$. Put
\begin{equation}\label{eq:Ad}
  A_d=\frac1{|\mathcal F_d|}\sum_{C\in\mathcal F_d}v_C ,
\end{equation}
and let $V=M_d/2M_d\simeq\F_2^{2d}$ carry the reduced quadratic form $\bar q(u,v)=u\cdot v\in\F_2$.

\begin{theorem}\label{thm:SW}
For every $d\geq1$,
\begin{equation}\label{eq:SW}
  \mathcal P_d[e_0]=\big(1+2^{1-d}\big)A_d+\big(4^{1-d}-2^{-d}\big)\,v_{T_d}.
\end{equation}
Equivalently, in the basis $\{e_c\}$,
\begin{equation}\label{eq:SW-basis}
  \mathcal P_d[e_0]=e_0+4^{1-d}\sum_{Q(c)=0}e_c+2^{-d}\sum_{v\in V}(-1)^{\bar q(v)}e_{2v}.
\end{equation}
\end{theorem}

Equation~\eqref{eq:SW} is the case $p=2$ of \cite[equations~(D.23)--(D.29)]{BDS25}, written for the
free orbit $a=d$ and solved for the handlebody sum. Indeed, at $p=2$ and $a=n=d$ the two coefficients of
\cite[equations~(D.28)--(D.29)]{BDS25} are $2^d/(2^d+2)$ and $(2^d-4)/(2^d+2)$, provided the second
multiplies the normalized torsion vector $2^{-d}v_{T_d}$ of~\eqref{eq:torsion-state}. Explicitly, the
free-orbit formula reads
\begin{equation}\label{eq:BDS-free}
  A_d=\frac{2^d}{2^d+2}\,\mathcal P_d[e_0]+\frac{2^d-4}{2^d+2}\,\big(2^{-d}v_{T_d}\big),
\end{equation}
and solving this equation for $\mathcal P_d[e_0]$ gives exactly~\eqref{eq:SW}. For completeness we give a direct proof in
appendix~\ref{app:handlebody}, by computing the six handlebody vectors in the Weil representation.
The term $e_0$ comes from the identity, the second term from the four handlebodies $S$, $TS$, $T^2S$
and $STS$, and the last term from the single handlebody $ST^2S$.

\subsection{The full genus-one sum}\label{sec:ens-full}

The full bulk sum contains the handlebodies and the torsion topology. Its form is fixed by
symmetry. The sum must be invariant under the modular group, since it is a sum over geometries
with a fixed boundary torus. It must also be invariant under the invertible bulk symmetries
$\Ot(M_d,\mathsf q)$, since the bulk sum does not single out any boundary condition. The space of
simultaneously invariant states is two-dimensional, spanned by the handlebody state
$\mathcal P_d[e_0]$ and the torsion state $v_{T_d}$ \cite{Dymarsky}. The uniform average over all codes also has both properties, so
it must be a combination of the two. 

\begin{proposition}\label{prop:full}
For every number of pairs $d$, the sum over all code CFTs is reproduced by the bulk sum over
handlebodies and the torsion topology,
\begin{equation}\label{eq:full-d}
  \sum_{C\in\mathcal L_d}Z_C=a_d\,Z_{\rm grav}+b_d\,Z_{T_d},
\end{equation}
or, at the level of states,
\begin{equation}\label{eq:full-d-vec}
  \sum_{C\in\mathcal L_d}v_C=a_d\,\mathcal P_d[e_0]+b_d\,v_{T_d},
\end{equation}
where the weight of the handlebodies is
\begin{equation}\label{eq:ad}
  a_d=4^{d-1}N_{d-1},
\end{equation}
and the weight of the torsion topology is
\begin{equation}\label{eq:bd}
  b_d=N_d-\big(4^{d-1}+1+2^{d-2}\big)N_{d-1}.
\end{equation}
Here $N_0=1$ and $N_d=|\mathcal L_d|$ is the finite sum~\eqref{eq:Nd}. Both weights are positive.
\end{proposition}

\begin{proof}
Write $W_d=\sum_{C\in\mathcal L_d}v_C$. Since the space of simultaneously invariant states is
two-dimensional,
\[
  W_d=a_d\,\mathcal P_d[e_0]+b_d\,v_{T_d}
\]
for some $a_d,b_d$. Choose a primitive isotropic label $c\in M_d$. Codes containing $c$ are in
bijection with the Lagrangians of $\langle c\rangle^\perp/\langle c\rangle\simeq M_{d-1}$, so the
coefficient of $e_c$ in $W_d$ is $N_{d-1}$. By~\eqref{eq:SW-basis} the coefficient of $e_c$ in
$\mathcal P_d[e_0]$ is $4^{1-d}$, and $c\notin T_d$. Hence $a_d4^{1-d}=N_{d-1}$, which
proves~\eqref{eq:ad}.

The coefficient of $e_0$ in $W_d$ is $N_d$, whereas its coefficient in $\mathcal P_d[e_0]$ is
$1+4^{1-d}+2^{-d}$. Consequently
\[
  N_d=a_d\big(1+4^{1-d}+2^{-d}\big)+b_d,
\]
which proves~\eqref{eq:bd}.

For $d=1$ the weights are $1$ and $\tfrac12$. For $d\geq2$, choose $v\in V$ with $\bar q(v)=1$ and
put $m_d=\#\{C\in\mathcal L_d:\ 2v\in C\}$. Since $T_d$ contains $2v$, $m_d\geq1$. The coefficient
of $e_{2v}$ in $\mathcal P_d[e_0]$ is $4^{1-d}-2^{-d}$, so comparing coefficients of $e_{2v}$ gives
\[
  b_d=m_d+\big(2^{d-2}-1\big)N_{d-1}>0. \qedhere
\]
\end{proof}

The first values are
\begin{equation}\label{eq:ab-table}
\begin{array}{rrrr}
  \toprule
  d & N_d & a_d & b_d\\
  \midrule
  1 & 3 & 1 & 1/2\\
  2 & 22 & 12 & 4\\
  3 & 486 & 352 & 68\\
  4 & 36766 & 31104 & 3232\\
  \bottomrule
\end{array}
\end{equation}
The formulas above use the unnormalized sum over codes
in~\eqref{eq:full-d}; for the uniform average, one needs to divide both sides by $N_d$. Moreover, in the
normalization $|\mathrm{tor}\rangle=2^{-d}v_{T_d}$ of the torsion state~\eqref{eq:torsion-state},
the coefficient of the torsion geometry is thus
\[
  W_d=a_d\,\mathcal P_d[e_0]+2^db_d\,|\mathrm{tor}\rangle,
\]
while the coefficient of the scalar partition function $Z_{T_d}$ in~\eqref{eq:full-d} remains $b_d$.

The weight of the handlebodies has a simple meaning. A handlebody sees a code through the
labels it contains, and each label of order four belongs to as many codes as there are codes with
one pair fewer. The torsion topology then supplies what the handlebodies miss. Its weight is
positive, so the full sum is always a genuine ensemble of CFTs, with positive weights, even when
the handlebodies alone are not (section~\ref{sec:ens-d3}). 

\subsubsection{Two pairs}\label{sec:ens-d2}

For $d=2$ the coefficient of $v_T$ in~\eqref{eq:SW} vanishes.

\begin{corollary}\label{cor:SW-d2}
For $d=2$ there are twelve free codes, four in each of the classes $F_0$, $F_1$, $F_2$, and
\begin{equation}\label{eq:Zgrav-d2}
  \mathcal P_2[e_0]=\frac18\sum_{C\in\mathcal F_2}v_C ,\qquad
  Z_{\rm grav}(\tau)=\frac12\big(Z_{F_0}(\tau)+Z_{F_1}(\tau)+Z_{F_2}(\tau)\big).
\end{equation}
\end{corollary}

The handlebody contribution is thus supported on a single orbit of $\Ot(M_2,\mathsf q)$, the free
orbit: the handlebody state admits a positive expansion using only free codes. To see how the
mixed codes nevertheless enter the full ensemble, let
$V_F$, $V_M$ and $V_T$ be the sums of the code vectors over the twelve free, the nine mixed and the
one torsion code. Their coefficients on the four relevant types of labels are
\begin{equation}\label{eq:orbit-coeffs}
\begin{array}{lcccc}
  \toprule
  & 0 & 2v\ \text{singular} & 2v\ \text{nonsingular} & c\ \text{primitive isotropic} \\
  \midrule
  V_F & 12 & 4 & 0 & 2\\
  V_M & 9 & 5 & 3 & 1\\
  V_T & 1 & 1 & 1 & 0\\
  \bottomrule
\end{array}
\end{equation}
A free code contains no label $2v$ with $\bar q(v)=1$, since its two-torsion is $2U$ with $U$ totally
singular. The table gives the relation $V_F-2V_M+6V_T=0$, which reflects the two-dimensionality of
the space of doubly invariant states.
Because of this relation, the expansion of the handlebody state in code states is not unique. The same vector also admits the signed expansion
\begin{equation}\label{eq:signed}
  \mathcal P_2[e_0]=\tfrac14V_M-\tfrac34V_T ,
\end{equation}
supported on the mixed and torsion codes. Of the two expansions, the positive one supported on
the free orbit is the natural one. Equivalently, the mixed-code sum satisfies
\begin{equation}\label{eq:VM}
  V_M=\tfrac12V_F+3V_T,
\end{equation}
so its contribution to the full ensemble is absorbed into both the handlebody and the torsion
terms. In particular $\sum_{C\in\mathcal L_2}v_C=V_F+V_M+V_T=12\,\mathcal P_2[e_0]+4\,v_T$, in
agreement with proposition~\ref{prop:full}.

 Let
\[
Z_{\mathrm{all}}=\frac{1}{22}\sum_{C\in\mathcal{L}_2}Z_C,
\qquad
Z_{\mathrm{free}}=\frac{1}{3}\bigl(Z_{F_0}+Z_{F_1}+Z_{F_2}\bigr).
\]
Pairing with $\Psi$ and dividing by twenty-two gives
\begin{equation}
\label{eq:full-boundary-average}
Z_{\mathrm{all}}
  =\frac{6}{11}Z_{\mathrm{grav}}+\frac{2}{11}Z_T
  =\frac{9}{11}Z_{\mathrm{free}}+\frac{2}{11}Z_T .
\end{equation}
This is the $d=2$ case of the TQFT gravity identity of~\cite{Dymarsky}, now written in terms of the
CFTs of the $A_3$ family. The first term is the sum over smooth handlebodies, which at $d=2$
reproduces the average over the free theories. The second term is the torsion topology, which
contributes the theory $T$ itself, with gap $1/2$ and current algebra $\widehat{\su}(2)_1^6$. The
mixed theories, including $M_2$ with the largest gap, have the same weight as all other boundaries
in $Z_{\rm all}$, but they do not appear as separate bulk contributions: by~\eqref{eq:VM} their
contribution is absorbed into both the handlebody and the torsion terms. Since $v_T$ is invariant, the same identity can be written as a single
modular sum with the seed of~\cite{Dymarsky},
\begin{equation}
\label{eq:full-boundary-seed}
Z_{\mathrm{all}}
  =\frac{6}{11}\left\langle\Psi,\,\mathcal{P}_2\!\left[e_0+\tfrac{1}{18}v_T\right]\right\rangle .
\end{equation}

\subsubsection{Three or more pairs}\label{sec:ens-d3}

For $d\geq3$ the coefficient $4^{1-d}-2^{-d}$ in~\eqref{eq:SW} is negative.

\begin{corollary}\label{cor:negative}
For $d\geq3$, the handlebody state $\mathcal P_d[e_0]$ cannot be expressed as a nonnegative linear
combination of the code boundary states $v_C$.
\end{corollary}

Choose $v\in V$ with $\bar q(v)=1$, for instance with first hyperbolic pair $(1,1)$ and all other
coordinates zero. By~\eqref{eq:SW-basis} the coefficient of $e_{2v}$ in $\mathcal P_d[e_0]$ is
$2^{-d}(2^{2-d}-1)<0$ for $d\geq3$. In any combination $\sum_C\lambda_Cv_C$ with $\lambda_C\geq0$
the coefficient of $e_{2v}$ is $\sum_{C\ni 2v}\lambda_C\geq0$.
The corollary concerns the bulk state. It does not by itself exclude a positive decomposition of
the scalar partition function $Z_{\rm grav}$, since the map from states to partition functions is
not injective; for instance $\Theta_c=\Theta_{-c}$.

This corollary illustrates, in the $A_3$ family, the statement of~\cite{Dymarsky} that the sum over
handlebodies alone is not a consistent bulk sum when $N$ is not square-free. For $d\geq3$ the
handlebody state is not even a nonnegative combination of boundary vectors. The full sum of
proposition~\ref{prop:full} is: the torsion topology, with its positive weight $b_d$, exactly
compensates the negative components. For $d=1$, by contrast, \eqref{eq:SW} reads
$\mathcal P_1[e_0]=v_{\{(a,0)\}}+v_{\{(0,b)\}}+\tfrac12v_{T_1}$, which is already positive and
reproduces \cite[equation~(2.15)]{Dymarsky}; it differs from the uniform average, which is
recovered with $a_1=1$, $b_1=\tfrac12$ by~\eqref{eq:full-d-vec}. These identities are genus-one statements. At genus $g$ the number of independent
doubly invariant states grows with $\min(d,g)$ \cite{Dymarsky}, and more classes of topologies
contribute.

\section{Discussion}\label{sec:disc}

The $A_3$ family is the simplest family of Narain code CFTs in which the alphabet is not
square-free and the codes are not all free. Most of
our results concern features that are absent for prime or square-free $N$, or that are hidden
in the generator-matrix description $(I\,|\,B)$ of~\cite{MO25}:
\begin{itemize}
    \item \textbf{The orbifold groupoid beyond free codes:}
The fact that two topological boundaries of an abelian Chern-Simons theory are related by
gauging is known~\cite{KS11,GK21}. Gaugings of $\Z_N$ symmetries of code CFTs were
studied in~\cite{Kawabata:2023iss,AKN26}, and shift orbifolds between Narain lattices were given
in~\cite{IMP26}. For the ($A_3, \Z_4$) family at $d=2$, we determined the structure of the groupoid with the order and the
isomorphism type of each gauge group following from the reductions of the two codes modulo two,
through proposition~\ref{prop:quotient-order}. It
replaces the computation of $C'/(C\cap C')$ by a statement about subspaces of the split
quadratic space $\F_2^4$.\\
From this, we identified two features that would not have appeared in a description restricted
to free codes. First, the torsion class $T$ is within one gauging of order two or four of every other class. It is reached
from every free class by a single $\Z_2\times\Z_2$ gauging, along which the full
$\widehat{\su}(2)_1^6$ current algebra arises in the twisted sector. Second, the three pairs $(F_1,M_0)$, $(F_0,M_2)$ and
$(F_2,M_2)$ are separated by gaugings of order eight only. The minimal gauging between two
boundaries is therefore not controlled by the bulk orbit structure alone: $F_1$ and $M_0$ lie in
different orbits, as do $F_1$ and $M_2$, but only the second pair is one $\Z_2$ gauging apart.
The minimal quantum dimension of section~\ref{sec:surf-action}, by contrast, is controlled by the
orbits alone.

\item \textbf{Invertible and non-invertible surfaces:}
We classified the topological surfaces of the bulk theory $\mathcal Z_2$. The invertible ones are
the T-dualities and fractional $B$-field shifts, and they act within the orbits of boundaries. In
the $A_3$ embedding the distinction between these two kinds of invertible surfaces becomes
physical: a T-duality preserves the Narain point and a $B$-field shift does not. Gaugings between orbits are non-invertible, and the minimal quantum
dimension of the surface that implements them depends only on how many $\Z_4$ factors of the code
change. This explains why pairs of theories that are related only by gaugings of high order are
nevertheless close in the bulk: a single condensation followed by a $B$-field shift relates them.
The non-invertible surfaces that implement gaugings include condensation defects, whose fusion we
determined, but also surfaces that combine a condensation with a nontrivial relabeling of the
transmitted lines. These are needed to leave the torsion boundary at minimal cost.

\item \textbf{Duality defects:}
The eight self-dualities of the web give non-invertible lines of Tambara-Yamagami type on five of
the seven theories, with the bicharacters~\eqref{eq:TY-table}. In the bulk they come from
invertible T-duality surfaces. This gives an explicit family of Narain CFTs with non-invertible
duality defects read off from the code. Such defects constrain the spectrum and the renormalization
group flows of the theory~\cite{CLSWY19}.

\item \textbf{TQFT gravity for the family:}
The genus-one sum over topologies of~\cite{Dymarsky,BDS25} has, for this family, a handlebody
contribution with weight $a_d=4^{d-1}N_{d-1}$ and a torsion contribution with positive weight
$b_d=N_d-(4^{d-1}+1+2^{d-2})N_{d-1}$. In the
$A_3$ embedding the torsion contribution is the partition function of the torsion theory. At $d=2$
the handlebody state admits a positive expansion using only free codes, and the uniform average is
$\tfrac9{11}Z_{\rm free}+\tfrac2{11}Z_T$. The mixed-code sum satisfies $V_M=\tfrac12V_F+3V_T$, so
its contribution is absorbed into both the handlebody and the torsion terms. For $d\geq3$ the
handlebody state alone is not a nonnegative combination of boundary vectors, which shows in this family why the sum must include the
torsion topology.
\end{itemize}There are several ways to extend the present work:
\begin{itemize}
    \item \textbf{Towards higher genus:} Our ensemble statements hold at genus one. The relation $V_F-2V_M+6V_T=0$ is a relation among genus-one code vectors, and it does not survive at higher genus, where the code vectors live in $\C[M_d^{\,g}]$ and the number of doubly invariant states grows with $\min(d,g)$ \cite{Dymarsky}. The mixed theories may then be produced by bulk configurations of their own. The general framework for such averages over the mapping class group is given in~\cite{Romaidis:2023zpx}. It would be interesting to identify which configurations of line defects in the sense of~\cite{Dymarsky,Angelinos:2025tqft} produce each orbit separately in the $A_3$ embedding.
\item \textbf{Fermionic theories:} Even though the positive lattice of $T$ is $\Z^{12}$, which is odd, it does not make $T$ a fermionic theory. For a vector $z=(x,y)$ of the Narain lattice, the scaling dimension and the spin of the corresponding state are $\Delta=\tfrac12|z|^2$ and $s=\tfrac12\eta(z,z)$. Oddness of the
positive form only means that $\Delta$ can be half-integral. The spin is controlled by the Lorentzian form, and the Narain lattice is even for this form, since $\eta(z,z)=2x\cdot y\in2\Z$.
Every local state of $T$ therefore has integral spin, and $T$ is a bosonic theory.  A fermionic theory differs in two ways. It contains local operators of half-integral spin, and its partition functions depend on a choice of spin structure. Obtaining a fermionic version of the present family would require a separate construction, in which a non-anomalous $\Z_2$ symmetry is specified as fermion parity and the
partition functions are defined on surfaces with spin structure~\cite{KTT19,GK16}. We leave this to future work.
\end{itemize}

\acknowledgments
This research is partially supported by an NSERC Discovery Grant.
\appendix
 
\section{Proofs for section~\ref{sec:web}}\label{app:web-proofs}

In this appendix we prove the results of section~\ref{sec:web}. Appendix~\ref{app:modular}
restates the twisted-sector results of~\cite{A3math} in terms of code labels.
Appendix~\ref{app:reduction} shows that the gauge group $C'/(C\cap C')$ is determined by the
reductions of $C$ and $C'$ modulo two. Appendix~\ref{app:web-proof} applies this to every pair
of classes, and appendix~\ref{app:pairs} lists an explicit pair of codes for each connection.
\subsection{Twisted sectors and the orbifold sum}\label{app:modular}

The statements of this subsection are~\cite[proposition~7.1]{A3math}, where they are proved with
lattice shift vectors. We restate them in terms of code labels, which is the form used in
section~\ref{sec:web}.

Let $C$, $C'$ be even self-dual, $I=C\cap C'$ and $H\simeq C'/I$ as in \eqref{eq:I-H}. Put
$f_{w,w'}=\mathbf 1_{C+w}\,\chi_{w'}$, so that $Z_C[w,w']=\Theta[f_{w,w'}]/|\eta|^{6d}$.

\begin{lemma}[{\cite[equation~(40)]{A3math}}]\label{lem:dual-quotients}
The pairing $(c+I,c'+I)\mapsto\beta(c,c')$ is perfect on $C/I\times C'/I$. Hence $C/I\simeq C'/I$.
\end{lemma}

\begin{proposition}\label{prop:twisted}
$Z_C[w,w']$ depends only on the classes of $w,w'$ in $C'/I$, and \eqref{eq:twisted-modular} holds.
\end{proposition}

\begin{proof}
This is the label form of the transformation laws in the proof of~\cite[proposition~7.1]{A3math}.
Changing $w$ or $w'$ by $i\in I$ changes neither $C+w$ nor $\chi_{w'}$ on $C+w$, since $\beta(C+w,i)=0$.
Under $\tau\mapsto\tau+1$, the block product with label $x=c+w$ acquires
$e^{2\pi i\,3q(x)/4}$~\cite[equation~(32)]{A3math}, and $q(x)=\beta(c,w)=\beta(x,w)$ because $q$ vanishes
on $C$ and $C'$; hence $f_{w,w'}\mapsto f_{w,w'+w}$. Under $\tau\mapsto-1/\tau$, the finite Fourier
transform of~\cite[equation~(33)]{A3math} sends $f_{w,w'}$ to $f_{w',-w}$, by character orthogonality
on $C$ and $\beta(w,w')=0$.
\end{proof}

\begin{proposition}[{\cite[equation~(42)]{A3math}}]\label{prop:orbifold}
Equation~\eqref{eq:orbifold-sum} holds, and the $H$-invariant states of the $w$-twisted sector have
labels $I+w$.
\end{proposition}

\begin{proof}
For $x\in C+w$, character orthogonality gives $|H|^{-1}\sum_{w'}\chi_{w'}(x)=\mathbf 1_{C'}(x)$. The
projected sector is $(C+w)\cap C'=I+w$, and these cosets exhaust $C'$.
\end{proof}
 
\subsection{Reduction modulo two}\label{app:reduction}
 
Let $V=M_2/2M_2\simeq\mathbb F_2^4$ with polar form $\bar\beta$, let $\pi$ be reduction, and put $W_C=\pi(C)$, $C[2]=C\cap2M_2$. By \cite[lemma~4.1, propositions~4.2 and~4.3]{A3math}: $W_C$ is a singular plane for $C\simeq\mathbb Z_4^2$, the line $\langle r\rangle$ for $C=C_r\simeq\mathbb Z_4\oplus\mathbb Z_2^2$, and $0$ for $T$. Each coordinate plane $\langle e(x),e(y)\rangle$ ($x\in\{a_1,b_1\}$, $y\in\{a_2,b_2\}$) carries one code of class $F_0$ and one of class $F_2$. Each diagonal plane $D_+=\langle1001,0110\rangle$, $D_-=\langle1100,0011\rangle$ carries two codes of class $F_1$. The classes $M_0,M_1,M_2$ are the $C_r$ with $r$ of weight one, $r=e(x)+e(y)$ with $x,y$ in different pairs, and $r=1111$.
 
\begin{lemma}\label{lem:torsion}
$C[2]=2W_C^\perp$.
\end{lemma}
 
\begin{proof}
If $2t\in C$, then $0=\beta(c,2t)\equiv2\bar\beta(c,t)\pmod 4$ for all $c\in C$, so $t\in W_C^\perp$. Both sides have order $16/|W_C|$.
\end{proof}
 
\begin{lemma}[{\cite[lemma~4.1]{A3math}}]\label{lem:rulings}
Identify $(a_1,a_2,b_1,b_2)$ with the matrix $\left(\begin{smallmatrix}a_1&a_2\\ b_2&b_1\end{smallmatrix}\right)\in M_2(\F_2)$,
so that $\bar q$ is the determinant. For a nonzero $u\in\F_2^2$, let $P_u$ be the plane of matrices with
column space in $\langle u\rangle$, and $Q_u$ the plane of matrices with row space in $\langle u\rangle$.
The six totally singular planes of $V$ are
\[
  P_{e_1}=\langle e(a_1),e(a_2)\rangle,\quad P_{e_2}=\langle e(b_1),e(b_2)\rangle,\quad P_{e_1+e_2}=D_+,
\]
\[
  Q_{e_1}=\langle e(a_1),e(b_2)\rangle,\quad Q_{e_2}=\langle e(a_2),e(b_1)\rangle,\quad Q_{e_1+e_2}=D_-.
\]
Two distinct column planes, or two distinct row planes, meet in $0$. A column plane $P_u$ and a row
plane $Q_v$ meet in the line spanned by $uv^{\mathsf T}$.
\end{lemma}

\begin{proof}
By~\cite[lemma~4.1]{A3math}, a nonzero singular vector is a rank-one matrix $uv^{\mathsf T}$ with $u,v$
uniquely determined, and the totally singular planes are the matrices with a fixed column line or a
fixed row line, three of each. A nonzero matrix in $P_u\cap P_{u'}$ would have column line
$\langle u\rangle=\langle u'\rangle$, so $u=u'$; the same holds for row planes. A nonzero matrix in
$P_u\cap Q_v$ has column line $\langle u\rangle$ and row line $\langle v\rangle$, hence equals $uv^{\mathsf T}$.
\end{proof}

\begin{proposition}\label{prop:quotient-order}
Let $C\neq C'$ be even self-dual codes, and put
\[
  I=C\cap C',\qquad G=C'/I,\qquad U=W_C,\qquad V=W_{C'},
\]
\[
  \omega=\dim U,\qquad \omega'=\dim V,\qquad k=\dim(U\cap V),\qquad \ell=\dim\pi(I),
\]
with $|G|=2^m$ and $\rho$ the rank of $\bar\beta$ restricted to $U\times V$. All dimensions and orthogonal
complements are taken over $\F_2$, with respect to $\bar\beta$.
\begin{enumerate}
  \item $|I|=2^{\,4-\omega-\omega'+k+\ell}$ and $|G|=2^{\,\omega+\omega'-k-\ell}$.
  \item If $U\neq V$, then $\ell=k$, so that $|G|=2^{\,\omega+\omega'-2k}$. If $U=V$, then both codes are free,
  $k=2$, $\ell=0$ and $|G|=4$.
  \item $G\simeq\Z_4^{\,\rho}\oplus\Z_2^{\,m-2\rho}$. In particular, when $|G|=4$, $G$ is cyclic exactly when
  $\bar\beta(U,V)\neq0$, and $G\simeq\Z_2\times\Z_2$ otherwise.
\end{enumerate}
\end{proposition}

\begin{proof}
By lemma~\ref{lem:torsion}, which is the torsion calculation
of~\cite[section~4.2]{A3math}, $C\cap2M_2=2U^\perp$ and $C'\cap2M_2=2V^\perp$. Hence
$I\cap2M_2=2(U^\perp\cap V^\perp)=2(U+V)^\perp$, of dimension $4-\omega-\omega'+k$, and reduction on $I$ gives
(1).

\emph{The case $U\neq V$.} Since $\pi(I)\subset U\cap V$, only $k=1$ needs proof. Write
$U\cap V=\langle w\rangle$ and choose lifts $c\in C$, $c'\in C'$ of $w$, with $c'-c=2t$. Evenness gives
$\bar\beta(w,t)=0$, so $t\in w^\perp=U^\perp+V^\perp$. Writing $t=u+v$ in this sum produces the common lift
$c+2u=c'-2v$. Thus $\ell=k$.

If $U=V$, the two codes cannot be mixed, since mixed codes are uniquely determined
by their reduction~\cite[proposition~4.2]{A3math}, and they are not both $2M_2$. So they are distinct free
codes over the same plane, and $k=2$. Distinct free lifts of one plane have no common nonzero residue.
Indeed, suppose $w\in U$ is nonzero with a common lift $c\in I$, and complete it to a basis $w,u$ of $U$
with lifts $d\in C$, $d'\in C'$ of $u$, so that $d'=d+2t$. Isotropy with $c$ gives $\bar\beta(w,t)=0$, and
evenness of $d$ and $d'$ gives $\bar\beta(u,t)=0$. Hence $t\in U^\perp$ and $d'\in d+2U^\perp\subset C$. Since
$C'$ is generated by $c$, $d'$ and $C'\cap2M_2=2U^\perp\subset C$, this forces $C'=C$. Therefore $\ell=0$ and
$|G|=4$. For example,
\[
  C=\langle\w{0001},\w{0010}\rangle,\qquad C'=\langle\w{0210},\w{2001}\rangle
\]
have $U=V=\langle\w{0001},\w{0010}\rangle$, but $I=\langle\w{0002},\w{0020}\rangle$ and $G\simeq\Z_2^2$, whereas
the formula $2^{\,\omega+\omega'-2k}$ would give one.

Doubling identifies $2C'$ with $C'/(C'\cap2M_2)\simeq V$, and
$2C'\cap I=2C'\cap(C\cap2M_2)$ corresponds to $V\cap U^\perp$. Hence
\[
  2G\simeq\frac{2C'}{2C'\cap I}\simeq\frac{V}{V\cap U^\perp},\qquad |2G|=2^{\rho}.
\]
Since $G$ has exponent dividing four, its elementary divisors are $G\simeq\Z_4^{\,\rho}\oplus\Z_2^{\,m-2\rho}$.
When $|G|=4$, $\rho=0$ exactly when $\bar\beta(U,V)=0$.
\end{proof}

For distinct mixed codes $C_r$, $C_{r'}$, the order is always four, and the criterion is simply
$\bar\beta(r,r')=1$ for $\Z_4$, or zero for $\Z_2^2$. For a free code with reduction $U$ and a mixed code with
reduction $\langle r\rangle$, one obtains
\[
  G\simeq\begin{cases}\Z_2, & r\in U,\\ \Z_4\oplus\Z_2, & r\notin U.\end{cases}
\]
For two free codes with different reductions, $G\simeq\Z_4$ when $k=1$ and $G\simeq\Z_4^2$ when $k=0$.
 
\subsection{Proof of the web}\label{app:web-proof}
 
We go through the reductions using proposition~\ref{prop:quotient-order}. Recall that $\bar\beta(e(x),e(y))=1$ iff $x,y$ are the two coordinates of one pair, and that $\bar\beta(s,1111)$ is the weight of $s$ modulo two.
 
\textbf{Free-free} ($\omega=\omega'=2$). Distinct planes in one ruling give order $16$. Planes in different rulings give order $4$, and the quotient is $\mathbb Z_4$, since $W_C\perp W_{C'}$ would force $W_{C'}\subset W_C^\perp=W_C$. Each ruling contains two coordinate planes and one diagonal plane, so this gives the $\mathbb Z_4$ connections $F_0F_1$, $F_0F_2$, $F_1F_2$ and the $\mathbb Z_4$ self-connections of $F_0$, $F_1$, $F_2$. Distinct codes over one plane give $\mathbb Z_2^2$: over a coordinate plane this is $F_0F_2$, over a diagonal plane the self-connection of $F_1$.
 
\textbf{Free-mixed} ($\omega=2$, $\omega'=1$). The quotient is $\mathbb Z_2$ if $r\in W_C$ and of order $8$ otherwise. A coordinate plane $\langle e(x),e(y)\rangle$ contains $e(x)$, $e(y)$ and $e(x)+e(y)$, but not $1111$. A diagonal plane contains two $M_1$-vectors and $1111$, but no weight-one vector. This gives the $\mathbb Z_2$ connections $F_0M_0$, $F_2M_0$, $F_0M_1$, $F_2M_1$, $F_1M_1$, $F_1M_2$, and order $8$ for every pair in $F_1M_0$, $F_0M_2$, $F_2M_2$.
 
\textbf{Free-$T$} ($\omega'=0$). Order $4$ with $W_C\perp0$: $\mathbb Z_2^2$ for $F_0$, $F_1$, $F_2$.
 
\textbf{Mixed-mixed.} Distinct mixed codes have $r\neq r'$, so $k=0$ and the order is $4$, with $\mathbb Z_2^2$ iff $\bar\beta(r,r')=0$. For $M_0M_0$, $\bar\beta(e(x),e(y))$ takes both values. For $M_0M_1$, $\bar\beta(e(x),e(y)+e(y'))=1$ iff the partner of $x$ is $y$ or $y'$, and both cases occur. For $M_1M_1$, $\bar\beta(e(x)+e(y),e(x')+e(y'))=[x\neq x']+[y\neq y']$ takes both values. For $M_0M_2$, $\bar\beta(e(x),1111)=1$, giving only $\mathbb Z_4$. For $M_1M_2$, $\bar\beta(r,1111)=0$, giving only $\mathbb Z_2^2$.
 
\textbf{Mixed-$T$} ($\omega=1$, $\omega'=0$): $\mathbb Z_2$ for $M_0$, $M_1$, $M_2$.
 
The classes $M_2$ and $T$ are single codes and have no self-connections.
 
\subsection{Representative pairs}
\label{app:pairs}

Table~\ref{tab:pairs} gives an explicit pair of codes $(C_X,\, g\,C_Y)$
for every connection and self-connection of table~\ref{tab:web}. Here
$C_X$ is the representative of class $X$ in \cite[table~1]{A3math}, and
$g$ is a block-monomial move applied to the representative $C_Y$; the
column ``generators of $C'$'' lists generators of $C' = g\,C_Y$. These are
the pairs used in table~\ref{tab:spectral}. The last column
gives $\beta(c,w)$ on the labels $c=(a_1,a_2,b_1,b_2)$ of $\Lambda(C)$, for
generators $w$ of $H = C'/I$; the gauged symmetry acts by
$\chi_w(c) = i^{3\beta(c,w)}$.
\newpage
\begin{table}[htbp]
\centering
\small
\renewcommand{\arraystretch}{1.12}
\begin{tabular}{lllll}
\hline
Pair & $H$ & $g$ & generators of $C' = g\,C_Y$ & generators of $I$ \quad\ $\beta(c,w)$ \\
\hline
$\mathrm F_0\mathrm F_0$ & $\mathbb Z_4$ & $\varepsilon_1$ & $0001,\ 1000$ & $0001$ \quad $w=1000$: $b_1$ \\
$\mathrm F_1\mathrm F_1$ & $\mathbb Z_4$ & $\varepsilon_1$ & $0130,\ 1001$ & $1311$ \quad $w=0130$: $3a_1+b_2$ \\
$\mathrm F_1\mathrm F_1$ & $\mathbb Z_2^2$ & $\nu_1$ & $0013,\ 1100$ & $0022,\ 2200$ \quad $0013$: $a_1+3a_2$;\ $1100$: $b_1+b_2$ \\
$\mathrm F_2\mathrm F_2$ & $\mathbb Z_4$ & $\varepsilon_1$ & $0120,\ 1002$ & $2120$ \quad $w=1002$: $2a_2+b_1$ \\
$\mathrm M_0\mathrm M_0$ & $\mathbb Z_4$ & $\varepsilon_2$ & $0100,\ 0020,\ 2000$ & $0020,\ 2000$ \quad $w=0100$: $b_2$ \\
$\mathrm M_0\mathrm M_0$ & $\mathbb Z_2^2$ & $\sigma$ & $0010,\ 0002,\ 0200$ & $0002,\ 0020$ \quad $0010$: $a_1$;\ $0200$: $2b_2$ \\
$\mathrm M_1\mathrm M_1$ & $\mathbb Z_4$ & $\varepsilon_1$ & $1001,\ 0002,\ 0220$ & $0002,\ 2220$ \quad $w=1001$: $a_2+b_1$ \\
$\mathrm M_1\mathrm M_1$ & $\mathbb Z_2^2$ & $\varepsilon_1\varepsilon_2$ & $1100,\ 0200,\ 0022$ & $0022,\ 2200$ \quad $0200$: $2b_2$;\ $1100$: $b_1+b_2$ \\
\hline
$\mathrm F_0\mathrm M_0$ & $\mathbb Z_2$ & $1$ & $0001,\ 0020,\ 2000$ & $0001,\ 0020$ \quad $w=2000$: $2b_1$ \\
$\mathrm F_0\mathrm M_1$ & $\mathbb Z_2$ & $1$ & $0011,\ 0002,\ 2200$ & $0011,\ 0002$ \quad $w=2200$: $2b_1+2b_2$ \\
$\mathrm F_1\mathrm M_1$ & $\mathbb Z_2$ & $1$ & $0011,\ 0002,\ 2200$ & $0011,\ 2200$ \quad $w=0002$: $2a_2$ \\
$\mathrm F_1\mathrm M_2$ & $\mathbb Z_2$ & $1$ & $1113,\ 0022,\ 0202$ & $1311,\ 0022$ \quad $w=0202$: $2a_2+2b_2$ \\
$\mathrm F_2\mathrm M_0$ & $\mathbb Z_2$ & $\sigma$ & $0010,\ 0002,\ 0200$ & $0012,\ 0200$ \quad $w=0002$: $2a_2$ \\
$\mathrm F_2\mathrm M_1$ & $\mathbb Z_2$ & $\varepsilon_2$ & $0110,\ 0020,\ 2002$ & $2112,\ 0020$ \quad $w=0110$: $a_1+b_2$ \\
$\mathrm M_0\mathrm T$ & $\mathbb Z_2$ & $1$ & $\mathrm T$ & $0002,\ 0020,\ 2000$ \quad $w=0200$: $2b_2$ \\
$\mathrm M_1\mathrm T$ & $\mathbb Z_2$ & $1$ & $\mathrm T$ & $0002,\ 0020,\ 2200$ \quad $w=0200$: $2b_2$ \\
$\mathrm M_2\mathrm T$ & $\mathbb Z_2$ & $1$ & $\mathrm T$ & $0022,\ 0202,\ 2002$ \quad $w=0002$: $2a_2$ \\
\hline
$\mathrm F_0\mathrm F_1$ & $\mathbb Z_4$ & $1$ & $0011,\ 1300$ & $0011$ \quad $w=1300$: $b_1+3b_2$ \\
$\mathrm F_0\mathrm F_2$ & $\mathbb Z_4$ & $1$ & $0012,\ 2100$ & $0012$ \quad $w=2100$: $2b_1+b_2$ \\
$\mathrm F_1\mathrm F_2$ & $\mathbb Z_4$ & $\varepsilon_1$ & $0120,\ 1002$ & $1322$ \quad $w=0120$: $2a_1+b_2$ \\
$\mathrm M_0\mathrm M_1$ & $\mathbb Z_4$ & $\varepsilon_2$ & $0110,\ 0020,\ 2002$ & $0020,\ 2002$ \quad $w=0110$: $a_1+b_2$ \\
$\mathrm M_0\mathrm M_2$ & $\mathbb Z_4$ & $1$ & $1113,\ 0022,\ 0202$ & $0022,\ 2002$ \quad $w=1113$: $a_1+3a_2+b_1+b_2$ \\
\hline
$\mathrm F_0\mathrm F_2$ & $\mathbb Z_2^2$ & $\varepsilon_2$ & $0210,\ 2001$ & $0002,\ 0020$ \quad $0210$: $a_1+2b_2$;\ $2001$: $a_2+2b_1$ \\
$\mathrm F_0\mathrm T$ & $\mathbb Z_2^2$ & $1$ & $\mathrm T$ & $0002,\ 0020$ \quad $0200$: $2b_2$;\ $2000$: $2b_1$ \\
$\mathrm F_1\mathrm T$ & $\mathbb Z_2^2$ & $1$ & $\mathrm T$ & $0022,\ 2200$ \quad $0002$: $2a_2$;\ $0200$: $2b_2$ \\
$\mathrm F_2\mathrm T$ & $\mathbb Z_2^2$ & $1$ & $\mathrm T$ & $0020,\ 0200$ \quad $0002$: $2a_2$;\ $2000$: $2b_1$ \\
$\mathrm M_0\mathrm M_1$ & $\mathbb Z_2^2$ & $1$ & $0011,\ 0002,\ 2200$ & $0002,\ 0020$ \quad $0011$: $a_1+a_2$;\ $2200$: $2b_1+2b_2$ \\
$\mathrm M_1\mathrm M_2$ & $\mathbb Z_2^2$ & $1$ & $1113,\ 0022,\ 0202$ & $0022,\ 2200$ \quad $0202$: $2a_2+2b_2$;\ $1113$: $a_1+3a_2+b_1+b_2$ \\
\hline
\end{tabular}
\caption{One explicit pair $(C_X,\, g\,C_Y)$ for every connection and
self-connection of table~\ref{tab:web}. Here $g$ is the block-monomial move applied to the representative $C_Y$:
$\sigma$ swaps the two pairs ($a_1\leftrightarrow a_2$,
$b_1\leftrightarrow b_2$), $\varepsilon_i$ swaps $a_i\leftrightarrow b_i$,
$\nu_i$ flips the sign of $(a_i,b_i)$, and $g=1$ means no move; each is a
T-duality, so $g\,C_Y$ is a code in the same class as $C_Y$.}
\label{tab:pairs}
\end{table}

\section{Surfaces of $\mathcal Z_2$}\label{app:surfaces}

\paragraph{Lifting Lagrangians.} Let $(M_n,\mathsf q)$ be the split module with $n$ pairs and let
$U\subset\F_2^{2n}$ be a totally singular subspace of dimension $k$ with basis $u_1,\dots,u_k$.
A Lagrangian $C$ with $C\bmod 2=U$ satisfies $C\cap2M_n=2U^\perp$, by the argument of
lemma~\ref{lem:torsion}, so $C$ is determined by lifts $c_i=\hat u_i+2t_i$ of the basis, with
$\hat u_i$ a fixed $\{0,1\}$-lift and $t_i\in\F_2^{2n}$ defined modulo $U^\perp$. Put
$\lambda_{ij}=\bar\beta(u_i,t_j)$. Evenness of $c_j$ fixes $\lambda_{jj}=\mathsf q(\hat u_j)/2$, and
isotropy of $c_i,c_j$ fixes $\lambda_{ij}+\lambda_{ji}=\beta(\hat u_i,\hat u_j)/2$ modulo two. The
entries $\lambda_{ij}$ with $i<j$ are free, and distinct choices give distinct Lagrangians. Hence
each $U$ has $2^{\binom k2}$ lifts. With the standard count of totally singular subspaces of the
split quadratic space, this gives~\eqref{eq:Nd}. 
This is the case $p=2$ of the count in \cite[appendix~D.1, equation~(D.9)]{BDS25}.

\paragraph{Bicharacters.} Let $C'=gC$ with $g$ a T-duality. The dual symmetry of $\Lambda(C')$ is
$C/I\subset M_2/C'$, and the duality $\Lambda(C')\simeq\Lambda(C)$ maps a symmetry label $w$ to $g^{-1}w$.
For the self-duality to be of Tambara-Yamagami type for $H=(C+C')/C\subset M_2/C$, the image of the
dual symmetry must be $H$, that is $g(C+C')=C+C'$. The
bicharacter is then $\chi(h,h')=\langle f^{-1}(h),h'\rangle$, with $f:C/I\to H$ induced by $g^{-1}$ and
$\langle\cdot,\cdot\rangle$ the perfect pairing of lemma~\ref{lem:dual-quotients}. Writing $gh=c+c'$ with
$c\in C$, $c'\in C'$, and using that $C'$ is isotropic, this gives~\eqref{eq:chi}. In all eight cases
$\chi$ is symmetric and nondegenerate, with the values listed in~\eqref{eq:TY-table}.

\section{Proof of theorem~\ref{thm:SW}}\label{app:handlebody}
The group $\mathrm{SL}_2(\Z_4)$ has order $48$, and the image of $\Gamma_\infty$ has order $8$. Six
representatives of the right cosets are $1,S,TS,T^2S,STS,ST^2S$. Since
$\rho(T)e_0=\rho(-I)e_0=e_0$, these representatives give the stated handlebody sum. Put
$\zeta(c)=e^{2\pi iQ(c)}$. The Weil formula gives
\[
  \rho(S)e_0=4^{-d}\sum_ce_c,\qquad \rho(TS)e_0=4^{-d}\sum_c\zeta(c)e_c,
\]
\[
  \rho(T^2S)e_0=4^{-d}\sum_c\zeta(c)^2e_c,\qquad \rho(STS)e_0=4^{-d}\sum_c\zeta(c)^{-1}e_c.
\]
For the fourth formula, one completes the square in the finite Fourier transform and uses
$\sum_{x\in M_d}e^{2\pi iQ(x)}=4^d$. As $\zeta(c)^4=1$, the four displayed vectors add up to
$4^{1-d}\sum_{Q(c)=0}e_c$.

The remaining Fourier transform is
\[
  \rho(ST^2S)e_0=4^{-2d}\sum_{c\in M_d}\Big(\sum_{x\in M_d}(-1)^{x_a\cdot x_b}e^{-2\pi iB(c,x)}\Big)e_c .
\]
It factorizes into $d$ two-dimensional sums. For one hyperbolic pair, the inner sum vanishes unless both
coordinates of $c$ are even; if $c=(2u,2v)$ it is $8(-1)^{uv}$. Consequently,
\[
  \rho(ST^2S)e_0=2^{-d}\sum_{v\in V}(-1)^{\bar q(v)}e_{2v}.
\]
Adding $e_0$ proves~\eqref{eq:SW-basis}.

To obtain~\eqref{eq:SW}, let $C\in\mathcal F_d$ and $U=C\bmod 2\subset V$. Since $C$ is free, $U$ is a maximal
totally singular $d$-space for $\bar q$, and $C[2]=2U$. Conversely, every maximal totally singular $U$ has exactly $2^{\binom d2}$ free lifts. Indeed,
after fixing one isotropic lift and a dual hyperbolic basis, the possible lifts are parametrized by matrices
$(\lambda_{ij})\in M_d(\F_2)$ satisfying $\lambda_{ii}=0$ and $\lambda_{ij}=\lambda_{ji}$. There are
$2^{\binom d2}$ such matrices. 

The number of nonzero singular vectors in $V$ is $(2^{d-1}+1)(2^d-1)$, while each maximal totally singular
$d$-space contains $2^d-1$ nonzero vectors. A fixed nonzero singular $v\in V$ therefore lies in the fraction
$(2^{d-1}+1)^{-1}$ of the maximal totally singular spaces. Equivalently,
\begin{equation}\label{eq:frac-singular}
  \frac{\#\{C\in\mathcal F_d:\ 2v\in C\}}{|\mathcal F_d|}=\frac1{2^{d-1}+1}\qquad(v\neq0,\ \bar q(v)=0).
\end{equation}
No free code contains $2v$ for $\bar q(v)=1$, since its $2$-torsion is $2U$ with $U$ totally singular.

Now fix a primitive $c\in M_d$ with $Q(c)=0$, and set $v=c\bmod 2\neq0$. For each maximal totally singular $U$
containing $v$, precisely $2^{\binom{d-1}2}$ of its $2^{\binom d2}$ free lifts contain $c$: specifying $c$
fixes the $d-1$ entries in one row of the alternating lift matrix. Therefore
\begin{equation}\label{eq:frac-primitive}
  \frac{\#\{C\in\mathcal F_d:\ c\in C\}}{|\mathcal F_d|}=\frac1{2^{d-1}(2^{d-1}+1)}.
\end{equation}
At $c=0$, the coefficient of $A_d$ is $1$; when $Q(c)\neq0$ it is zero. Equations~\eqref{eq:frac-singular}
and~\eqref{eq:frac-primitive} therefore determine every coefficient of $A_d$. Comparing them
with~\eqref{eq:SW-basis} proves~\eqref{eq:SW}.



\begin{thebibliography}{99}


\bibitem{A3math}
R.~Sammani and A.~Sebbar, \emph{The $A_3$ Narain family over $\Z_4$},
\href{https://arxiv.org/abs/2610.26610}{\texttt{arXiv:2610.26610}}
\bibitem{DS21}
A.~Dymarsky and A.~Shapere, \emph{Quantum stabilizer codes, lattices, and CFTs}, \href{https://doi.org/10.1007/JHEP03(2021)160}{\emph{JHEP} \textbf{03} (2021) 160} [\href{https://arxiv.org/abs/2009.01244}{\texttt{arXiv:2009.01244}}].

\bibitem{Alam23}
Y.~F.~Alam, K.~Kawabata, T.~Nishioka, T.~Okuda and S.~Yahagi, \emph{Narain CFTs from nonbinary stabilizer codes}, \href{https://doi.org/10.1007/JHEP12(2023)127}{\emph{JHEP} \textbf{12} (2023) 127} [\href{https://arxiv.org/abs/2307.10581}{\texttt{arXiv:2307.10581}}].

\bibitem{MO25}
S.~Mizoguchi and T.~Oikawa, \emph{Unifying error-correcting code/Narain CFT correspondences via lattices over integers of cyclotomic fields}, \emph{Phys.\ Lett.\ B} \textbf{862} (2025) 139308 [\href{https://arxiv.org/abs/2410.12488}{\texttt{arXiv:2410.12488}}].

\bibitem{SS26}
E.~H.~Saidi and R.~Sammani, \emph{Algebraic constructions of code lattices in Narain CFTs}, \emph{J.\ Math.\ Phys.} \textbf{67} (2026) 032301 [\href{https://arxiv.org/abs/2605.03643}{\texttt{arXiv:2605.03643}}].

\bibitem{Ang25}
N.~Angelinos, \emph{Code construction and ensemble holography of simply-laced WZW models at level~1}, \href{https://doi.org/10.1007/JHEP06(2025)002}{\emph{JHEP} \textbf{06} (2025) 002} [\href{https://arxiv.org/abs/2503.04055}{\texttt{arXiv:2503.04055}}].

\bibitem{Kawabata:2023iss}
K.~Kawabata, T.~Nishioka and T.~Okuda,
\emph{Narain CFTs from quantum codes and their $\mathbb{Z}_2$ gauging},
\href{https://doi.org/10.1007/JHEP05(2024)133}{\emph{JHEP} \textbf{05} (2024) 133}
[\href{https://arxiv.org/abs/2308.01579}{\texttt{arXiv:2308.01579}}].

\bibitem{AKN26}
K.~Ando, K.~Kawabata and T.~Nishioka, \emph{Gauging $\mathbb{Z}_N$ symmetries of Narain CFTs}, \href{https://doi.org/10.1007/JHEP02(2026)058}{\emph{JHEP} \textbf{02} (2026) 058} [\href{https://arxiv.org/abs/2503.21524}{\texttt{arXiv:2503.21524}}].

\bibitem{IMP26}
D.~Isra\"el, I.~V.~Melnikov and Y.~Proto, \emph{Shift orbifolds, decompactification limits, and lattices}, \href{https://doi.org/10.21468/SciPostPhys.20.2.039}{\emph{SciPost Phys.} \textbf{20} (2026) 039} [\href{https://arxiv.org/abs/2502.18453}{\texttt{arXiv:2502.18453}}].

\bibitem{MW20}
A.~Maloney and E.~Witten, \emph{Averaging over Narain moduli space}, \href{https://doi.org/10.1007/JHEP10(2020)187}{\emph{JHEP} \textbf{10} (2020) 187} [\href{https://arxiv.org/abs/2006.04855}{\texttt{arXiv:2006.04855}}].

\bibitem{AJCHT21}
N.~Afkhami-Jeddi, H.~Cohn, T.~Hartman and A.~Tajdini, \emph{Free partition functions and an averaged holographic duality}, \href{https://doi.org/10.1007/JHEP01(2021)130}{\emph{JHEP} \textbf{01} (2021) 130} [\href{https://arxiv.org/abs/2006.04839}{\texttt{arXiv:2006.04839}}].

\bibitem{Rae21}
J.~Raeymaekers, \emph{A note on ensemble holography for rational tori}, \href{https://doi.org/10.1007/JHEP12(2021)177}{\emph{JHEP} \textbf{12} (2021) 177} [\href{https://arxiv.org/abs/2110.08833}{\texttt{arXiv:2110.08833}}].

\bibitem{ADS23}
O.~Aharony, A.~Dymarsky and A.~D.~Shapere, \emph{Holographic description of Narain CFTs and their code-based ensembles}, \href{https://doi.org/10.1007/JHEP05(2024)343}{\emph{JHEP} \textbf{05} (2024) 343} [\href{https://arxiv.org/abs/2310.06012}{\texttt{arXiv:2310.06012}}].

\bibitem{Barbar:2023ncl}
A.~Barbar, A.~Dymarsky and A.D.~Shapere,
\emph{Global symmetries, code ensembles, and sums over geometries},
\href{https://doi.org/10.1103/PhysRevLett.134.151603}{\emph{Phys. Rev. Lett.} \textbf{134} (2025) 151603}
[\href{https://arxiv.org/abs/2310.13044}{\texttt{arXiv:2310.13044}}].

\bibitem{Dymarsky}
A.~Dymarsky and A.~Shapere,
\emph{TQFT gravity and ensemble holography},
\href{https://doi.org/10.1007/JHEP02(2025)091}{\emph{JHEP} \textbf{02} (2025) 091}
[\href{https://arxiv.org/abs/2405.20366}{\texttt{arXiv:2405.20366}}].

\bibitem{BDS25}
A.~Barbar, A.~Dymarsky and A.~D.~Shapere, \emph{Holographic description of 4d Maxwell theories and their code-based ensembles}, [\href{https://arxiv.org/abs/2510.03392}{\texttt{arXiv:2510.03392}}].

\bibitem{DHM25}
A.~Dymarsky, J.~Henriksson and B.~McPeak, \emph{Holographic duality from Howe duality: Chern--Simons gravity as an ensemble of code CFTs}, \href{https://doi.org/10.1007/JHEP02(2026)257}{\emph{JHEP} \textbf{02} (2026) 257} [\href{https://arxiv.org/abs/2504.08724}{\texttt{arXiv:2504.08724}}].

\bibitem{Bar25}
A.~Barbar, \emph{Automorphism-weighted ensembles from TQFT gravity}, [\href{https://arxiv.org/abs/2511.04311}{\texttt{arXiv:2511.04311}}].

\bibitem{Angelinos:2025tqft}
N.~Angelinos,
\emph{Abelian 3D TQFT gravity, ensemble holography and stabilizer states},
\href{https://arxiv.org/abs/2509.26052}{\texttt{arXiv:2509.26052}}.


\bibitem{Saidi:2024zam}
E.H.~Saidi and R.~Sammani,
\emph{Algebraic realisation of the Zamolodchikov metric in Narain theories},
\href{https://doi.org/10.1002/prop.70098}{\emph{Fortsch. Phys.} (2026)}
[\href{https://arxiv.org/abs/2412.13932}{\texttt{arXiv:2412.13932}}].

\bibitem{Saidi:2026code}
E.H.~Saidi and R.~Sammani,
\emph{Code CFTs and topological matter},
\href{https://doi.org/10.1007/JHEP02(2026)007}{\emph{JHEP} \textbf{02} (2026) 007}
[\href{https://arxiv.org/abs/2506.22088}{\texttt{arXiv:2506.22088}}].


\bibitem{KS11}
A.~Kapustin and N.~Saulina, \emph{Topological boundary conditions in abelian Chern--Simons theory}, \emph{Nucl.\ Phys.\ B} \textbf{845} (2011) 393 [\href{https://arxiv.org/abs/1008.0654}{\texttt{arXiv:1008.0654}}].

\bibitem{KS-surf}
A.~Kapustin and N.~Saulina, \emph{Surface operators in 3d topological field theory and 2d rational conformal field theory}, in \emph{Mathematical Foundations of Quantum Field Theory and Perturbative String Theory}, Proc.\ Symp.\ Pure Math.\ \textbf{83} (2011) 175 [\href{https://arxiv.org/abs/1012.0911}{\texttt{arXiv:1012.0911}}].

\bibitem{FSV}
J.~Fuchs, C.~Schweigert and A.~Valentino, \emph{Bicategories for boundary conditions and for surface defects in 3-d TFT}, \emph{Commun.\ Math.\ Phys.} \textbf{321} (2013) 543 [\href{https://arxiv.org/abs/1203.4568}{\texttt{arXiv:1203.4568}}].

\bibitem{BDR21}
M.~Buican, A.~Dymarsky and R.~Radhakrishnan, \emph{Quantum codes, CFTs, and defects},
\emph{JHEP} \textbf{03} (2023) 017 [\href{https://arxiv.org/abs/2112.12162}{\texttt{arXiv:2112.12162}}].

\bibitem{BR23}
M.~Buican and R.~Radhakrishnan, \emph{Qudit stabilizer codes, CFTs, and topological surfaces},
[\href{https://arxiv.org/abs/2311.13680}{\texttt{arXiv:2311.13680}}].

\bibitem{NRS22}
P.~Niro, K.~Roumpedakis and O.~Sela, \emph{Exploring non-invertible symmetries in free theories},
\emph{JHEP} \textbf{03} (2023) 005 [\href{https://arxiv.org/abs/2209.11166}{\texttt{arXiv:2209.11166}}].

\bibitem{ACGHP24}
R.~Argurio, A.~Collinucci, G.~Galati, O.~Hulik and E.~Paznokas, \emph{Non-invertible T-duality at
any radius via non-compact SymTFT}, [\href{https://arxiv.org/abs/2409.11822}{\texttt{arXiv:2409.11822}}].

\bibitem{GK21}
D.~Gaiotto and J.~Kulp, \emph{Orbifold groupoids}, \href{https://doi.org/10.1007/JHEP02(2021)132}{\emph{JHEP} \textbf{02} (2021) 132} [\href{https://arxiv.org/abs/2008.05960}{\texttt{arXiv:2008.05960}}].

\bibitem{KOZ22}
J.~Kaidi, K.~Ohmori and Y.~Zheng, \emph{Symmetry TFTs for non-invertible defects}, [\href{https://arxiv.org/abs/2209.11062}{\texttt{arXiv:2209.11062}}].

\bibitem{FMT22}
D.~S.~Freed, G.~W.~Moore and C.~Teleman, \emph{Topological symmetry in quantum field theory}, \href{https://doi.org/10.4171/qt/223}{\emph{Quantum Topol.} \textbf{15} (2024) 779} [\href{https://arxiv.org/abs/2209.07471}{\texttt{arXiv:2209.07471}}].

\bibitem{CS99}
J.~H.~Conway and N.~J.~A.~Sloane, \emph{Sphere packings, lattices and groups}, 3rd ed., Springer, New York (1999).

\bibitem{RSS22}
K.~Roumpedakis, S.~Seifnashri and S.-H.~Shao, \emph{Higher gauging and non-invertible condensation
defects}, \emph{Commun.\ Math.\ Phys.} \textbf{401} (2023) 3043
[\href{https://arxiv.org/abs/2204.02407}{\texttt{arXiv:2204.02407}}].

\bibitem{weil}
F.~Str\"omberg, \emph{Weil representations associated with finite quadratic modules}, \emph{Math.\ Z.} \textbf{275} (2013) 509--527
[\href{https://arxiv.org/abs/1108.0202}{\texttt{arXiv:1108.0202}}].


\bibitem{Romaidis:2023zpx}
I.~Romaidis and I.~Runkel,
\emph{CFT correlators and mapping class group averages},
\href{https://arxiv.org/abs/2309.14000}{\texttt{arXiv:2309.14000}}.

\bibitem{TY98}
D.~Tambara and S.~Yamagami, \emph{Tensor categories with fusion rules of self-duality for finite abelian groups}, \emph{J.\ Algebra} \textbf{209} (1998) 692.

\bibitem{CLSWY19}
C.-M.~Chang, Y.-H.~Lin, S.-H.~Shao, Y.~Wang and X.~Yin, \emph{Topological defect lines and renormalization group flows in two dimensions}, \href{https://doi.org/10.1007/JHEP01(2019)026}{\emph{JHEP} \textbf{01} (2019) 026} [\href{https://arxiv.org/abs/1802.04445}{\texttt{arXiv:1802.04445}}].


\bibitem{KTT19}
A.~Karch, D.~Tong and C.~Turner, \emph{A web of 2d dualities: $\mathbb{Z}_2$ gauge fields and Arf invariants}, \href{https://doi.org/10.21468/SciPostPhys.7.1.007}{\emph{SciPost Phys.} \textbf{7} (2019) 007} [\href{https://arxiv.org/abs/1902.05550}{\texttt{arXiv:1902.05550}}].

\bibitem{GK16}
D.~Gaiotto and A.~Kapustin, \emph{Spin TQFTs and fermionic phases of matter}, \emph{Int.\ J.\ Mod.\ Phys.\ A} \textbf{31} (2016) 1645044 [\href{https://arxiv.org/abs/1505.05856}{\texttt{arXiv:1505.05856}}].



\end{thebibliography}
\end{document}